\documentclass[final,onefignum,onetabnum]{siamart171218}

\usepackage{lipsum}
\usepackage{amsfonts}
\usepackage{graphicx}
\usepackage{epstopdf}
\usepackage{algorithmic}
\ifpdf
  \DeclareGraphicsExtensions{.eps,.pdf,.png,.jpg}
\else
  \DeclareGraphicsExtensions{.eps}
\fi

\newsiamremark{remark}{Remark}
\newsiamremark{hypothesis}{Hypothesis}
\crefname{hypothesis}{Hypothesis}{Hypotheses}
\newsiamthm{claim}{Claim}

\headers{On DC based Methods for Phase Retrieval}{Meng Huang, Ming-Jun Lai, Abraham Varghese, and Zhiqiang Xu}

\title{On DC based Methods for Phase Retrieval}

\author{Meng Huang \thanks{Institute of Computational Mathematics, Academy of Mathematics and
Systems of Science, Chinese Academic Sciences, Beijing 100190, China. (\email{hm@lsec.cc.ac.cn}).}
\and Ming-Jun Lai \thanks{This author is partly supported by the National Science Foundation under grant DMS-1521537. Department of Mathematics, The University of Georgia, Athens, GA 30602. (\email{mjlai@uga.edu}).}
\and Abraham Varghese \thanks{Department of Mathematics, Shenandoah University, Winchester, VA 22601. (\email{opticsabru@gmail.com}).}
\and Zhiqiang Xu \thanks{Zhiqiang Xu was supported  by NSFC grant (11422113, 91630203, 11331012) and by National Basic Research Program of China (973 Program 2015CB856000).
Institute of Computational Mathematics, Academy of Mathematics and Systems of Science, Chinese Academic Sciences, Beijing 100190, China. (\email{xuzq@lsec.cc.ac.cn}).}
}

\usepackage{amsopn}

\ifpdf
\hypersetup{
  pdftitle={On DC based Methods for Phase Retrieval},
  pdfauthor={Meng Huang, Ming-Jun Lai, Abraham Varghese, and Zhiqiang Xu}
}
\fi

\def\bfa{{\bf a}}

\def\bfb{{\bf b}}

\def\bfc{{\bf c}}

\def\bfe{{\bf e}}

\def\bfu{{\bf u}}
\def\bfv{{\bf v}}
\def\bfw{{\bf w}}

\def\bfs{{\bf s}}

\def\bfx{{\bf x}}
\def\bfy{{\bf y}}
\def\bfz{{\bf z}}
\def\bfi{{\bf i}}

\newcommand{\abs}[1]{\left|#1\right|}

\def\3bar{{|\hspace{-.02in}|\hspace{-.02in}|}}

\def\RR{\mathbb{R}}
\newcommand{\C}{{\mathbb C}}

\newcommand{\normm}[1]{\|{#1}\|_2}
\newcommand{\T}{\top}

\newsiamthm{notation}{Notation}
\newsiamremark{example}{Example}

\begin{document}

\maketitle

% REQUIRED
\begin{abstract}
  In this paper, we develop  a new computational approach which is based on minimizing the
difference of two convex functionals (DC) to solve a broader class of phase retrieval problems.
The approach splits a standard nonlinear least squares minimizing function associated with the phase retrieval problem into the difference of two convex functions and then solves a sequence of convex minimization
sub-problems.   For each subproblem, the Nesterov's accelerated gradient descent algorithm or
the Barzilai-Borwein (BB) algorithm is used.  In the setting of sparse phase retrieval, a standard $\ell_1$
norm term is added into the  minimization mentioned above. The subproblem is approximated by a proximal
gradient method which is solved by the shrinkage-threshold technique directly without iterations. In addition,
a modified Attouch-Peypouquet technique is used to accelerate the iterative computation. These lead to
more effective algorithms than the  Wirtinger flow (WF) algorithm and the Gauss-Newton
(GN) algorithm and etc.. A convergence analysis of both DC based algorithms shows that the iterative
solutions is convergent linearly to a critical point
and will be closer to a global minimizer than the given initial starting point.
Our study is a deterministic analysis while
the study for the Wirtinger flow (WF) algorithm and its variants, the Gauss-Newton (GN) algorithm,
the trust region algorithm is based on the probability analysis.
 In particular, the DC based algorithms are able to retrieve solutions using a number $m$ of measurements
which is about twice of the number $n$ of entries in the solution with high frequency of successes.
When $m\approx n$, the $\ell_1$ DC based algorithm is able to retrieve sparse signals.
Finally, the paper discusses the nonexistence of the solution to the exact recovery of the phase retrieval
problem for arbitrary given measurement values. In addition, for a given set of measurement values, if there
is a solution, an estimate of the upper bound of the number of distinct solutions is also given.
\end{abstract}

% REQUIRED

% REQUIRED

\section{Introduction}
\subsection{Phase retrieval}
The phase retrieval problem has been extensively studied in the last 40 years due to its numerous applications,
such as X-ray diffraction, crystallography, electron microscopy, optical imaging and etc..
See, e.g. \cite{CLS15a},  \cite{MIJALH02},  \cite{M90}, \cite{R93}, \cite{SECCMS15}, \cite{F82}, and
\cite{DF87}. In particular, see \cite{JEH15} for an explanation of the image recovery from the phaseless
measurements and a survey of recent research results.
Mathematically, the phase retrieval problem or simply called phase retrieval problem can be
stated as follows.  Given measurement vectors
$\bfa_i\in \mathbb{R}^n$ (or $\in \mathbb{C}^n$), and the measurement values $b_i\geq 0$,
we would like to recover an unknown signal  $\bfx \in \mathbb{R}^n$ (or $\in \mathbb{C}^n$)
through a set of quadratic equations:
\begin{equation}
\label{quadcond}
b_1 =|\langle {\bfa}_1, \bfx\rangle|^2, \ldots, b_m = |\langle {\bfa}_m, \bfx\rangle|^2.
\end{equation}
Note that for any $c\in \mathbb{R}^n$ (or $\in \mathbb{C}^n$)  with $\abs{c}=1$,
we have $b_i=|\langle {\bfa}_i, c\bfx\rangle|^2$ for all $i$. Thus we can only hope to recover $\bfx$
up to a unimodular constant. One fundamental  problem in phase retrieval is to
give the minimal $m$ for which there exists $A=(\bfa_1,\ldots, \bfa_m)^\T$ which can recover
$\bfx$ up to a unimodular constant.
For the real case, it is well known that the minimal measurement number  $m$ is  $2n-1$
\cite{balan2006signal}. For the complex case $\mathbb{C}^n$, this question remains open.
Conca, Edidin, Hering and Vinzant \cite{conca2015algebraic} proved $m\ge 4n-4 $ generic measurement vectors
$F=(\mathbf{f}_1,\ldots,\mathbf{f}_m)^\T$ have phase retrieval property for $\mathbb{C}^n$ and
they furthermore show that $4n-4$  is sharp  if $n$ is in the form of $2^k+1,\;k\in \mathbb{Z}_+$.
In \cite{vinzant2015small}, for the case $n=4$, Vinzant  present $11=4n-5<4n-4$ measurement vectors
which have phase retrieval property for $\mathbb{C}^4$ which
implies that $4n-4$ is not sharp for some dimension $n$. There is a similar study for the sparse
phase retrieval. See \cite{WX14}.

There  are many computational algorithms available to find a true signal $\bfx$ up to a phase factor.
It is common folklore that for given $\bfa_i, i=1, \ldots, m$, we may not be able to find a solution
$\bfx$ for any given vector $\bfb=(b_1, \cdots, b_m)^\top$,
e.g. a perturbation of the exact observed value vector $\bfb^*$.
We shall give this fact a mathematical explanation (see \cref{nochance} in a later section). Thus,
the phase retrieval problem is usually formulated as follows:
\begin{equation}
\label{LSmin}
\min_{\bfx \in \mathbb{R}^n \hbox{ or } \mathbb{C}^n} \sum_{i=1}^m (|\langle {\bfa}_i, \bfx\rangle|^2 -b_i)^2.
\end{equation}
Although it is not a convex minimization problem, the minimizing functional is a nice differentiable function
and hence, many computational algorithms can be developed and they are very successful actually.
A gradient descent method (called  Wirtinger flow in the complex variable setting) is
developed by Cand\`es et al. in \cite{WF}. They show that the WF algorithm
converges to the true signal up to a global phase factor with high probability provided the measurement
vectors are $m=O(n\log n)$ Gaussian measurements.  Many variants of
Wirtinger flow algorithms were developed. See \cite{CLS15a}, \cite{CLM15}, \cite{TWF},
\cite{ZL16}, and  \cite{BSRW18} for Truncated WF, Thresholded WF, Reshaped WF, Accelerated WF, and etc..
In \cite{Gaoxu}, Gao and Xu propose a Gauss-Newton (GN) algorithm to find a minimizer of \cref{LSmin}.
They proved  that, for the real signal, the GN algorithm can converge to the global optimal solution
quadratically with $O(n\log n)$ measurements starting from a good initial solution.
Indeed, Gao and Xu provided a formula for the initialization which is
much better than the initialization in \cite{WF} in numerical experiments.
Another approach for the minimization \cref{LSmin} is called the true region method which was studied
in \cite{SQW17} where a geometric description of the landscape function
$f(\bfx) = \sum_{i=1}^m (|\langle {\bfa}_i, \bfx\rangle|^2 -b_i)^2$ is given.
To recover sparse signals from the measurements (\ref{quadcond}), a standard approach is to add $\ell_1$
term $\lambda \|\bfx\|_1$ to \cref{LSmin} or use the projected gradient method as discussed in \cite{SBE14}.

\subsection{Our contribution}
In this paper, we consider a broader class of phase retrieval problem which includes standard phase
retrieval as a special case. We aim to recover $\bfx \in \mathbb{R}^n$ (or $\in \mathbb{C}^n$)
from nonlinear measurements
\begin{equation}\label{problem setup}
b_i=f(\langle {\bfa}_i, \bfx\rangle),\quad i=1, \ldots, m,
\end{equation}
where $f:\mathbb{C}\to \mathbb{R}_+$ is a continuous convex function which satisfies the following
coercive condition:
\begin{equation*}
 f(x) \to \infty \hbox{ when } |x| \to \infty.
\end{equation*}
If we take $f(x)=|x|^2$, then it reduces to the standard phase retrieval.
To guarantee the unique recovery of $\bfx$, it has been
proved that the measurement number satisfies  $m \ge n+1$ for the real case ($2n+1$ for the complex
case, respectively) (see \cite{HX}). Recovering $\bfx$ from the nonlinear observation is also
raised in many areas, such as  neural network etc. (cf. \cite{Non1, Non2}).

To reconstruct $\bfx$ by solving \cref{problem setup}, we can formulate it as
\begin{equation} \label{main problem}
\min_{\bfx \in \mathbb{R}^n \hbox{ or } \mathbb{C}^n} \sum_{i=1}^m (f(\langle {\bfa}_i, \bfx\rangle) -b_i)^2.
\end{equation}
We approach it by using the standard technique for a difference of convex minimizing functionals.
Indeed, for the case $\bfx\in
\RR^n$ and $\bfa_i \in \RR^n$, let $F(\bfx)= \sum_{i=1}^m (f(\langle {\bfa}_i, \bfx\rangle) -b_i)^2$
be the minimizing functional. As it is not convex, we can write it as
$$
F(\bfx) = F_1(\bfx) - F_2(\bfx):= \sum_{i=1}^m \left( f^2(\langle {\bfa}_i, \bfx\rangle) + b_i^2 \right)
 - \sum_{i=1}^m 2 b_i f(\langle {\bfa}_i, \bfx\rangle).
$$
Note that $f$ is a convex function with function value $f(x)\ge 0$ for all $x \in \RR$.
Then $F_1$ and $F_2$ are convex functions. The minimization \cref{main problem} will be approximated by
\begin{equation}
\label{DCalg}
\bfx^{(k+1)}:= \hbox{arg} \min_\bfx  F_1(\bfx) - \nabla F_2(\bfx)^\top (\bfx- \bfx^{(k)})
\end{equation}
for any given $\bfx^{(k)}$. We call this algorithm as DC based algorithm following from the ideas in
\cite{G13}, where the sparse solutions of underdetermined linear system were studied.
Due to the nice properties of $F_1$ and $F_2$, we will be able to establish much better results
than those in \cite{G13}.
When $\bfx\in \mathbb{C}^n$ and ${\bfa}_j\in \mathbb{C}^n$,
$j=1, \ldots, m$, we have to write $\bfx= \bfx_R+ \sqrt{-1}\bfx_I$ and similar for ${\bfa}_j$. Letting
$\bfc=[\bfx_R^\top\,\,
\bfx_I^\top]^\top \in \mathbb{R}^{2n}$, we view $F_1(\bfx)$ as  a functions in $G_1(\bfc)
=F_1(\bfx_R+\sqrt{-1}\bfx_I)$. Then $G_1(\bfc)$ is a convex function of variable $\bfc$.
Similarly, $G_2(\bfc)=F_2(\bfx_R+ \sqrt{-1}\bfx_I)$ is a convex function of $\bfc$.
For convenience, we simply discuss the case when $\bfx,\;  {\bfa}_j, \; j=1,\ldots, m$ are real.  The complex variable setting can be treated in the same fashion.

The above minimization \cref{DCalg} is a convex minimization problem with differentiable functional
for each $k$.
We can solve it by using the standard gradient descent method with Nestrov's acceleration (cf. \cite{N04})
or by using the Barzilai-Borwein (BB) method (cf. \cite{BB88}).
There are several nice properties of this DC based approach. We can
show that
\[
F(\bfx^{(k+1)}) \le F(\bfx^{(k)}) - \ell \|\bfx^{(k+1)}- \bfx^{(k)}\|^2
\]
for some constant
$\ell>0$. That is, $F(\bfx^{(k)}), k\ge 1$
is strictly decreasing sequence and hence, the sequence $\bfx^{(k)}$ will not converge to a local maximum.
Also, we can prove the sequence $\{\bfx^{(k)}\}_{k=1}^\infty $
converges to a critical point $\bfx^*$. Using the Kurdyka-\L ojasiewicz inequality, we can also show
$\|\bfx^{(k)}- \bfx^*\|\le C \theta^k$ for  $\theta\in
(0, 1)$. If the function $F(\bfx)$ has the property that any global minimizer $\bfx^\star$ is a local minimizer
over a neighborhood $N(\bfx^\star)$ and the initial point $\bfx^{(1)}$ is within $N(\bfx^\star)$,
then the DC based algorithm will converge
to the global  minimizer. Actually, the function $F(\bfx)$ has such property for standard phase retrieval
problem and the initial point is chosen by a careful initialization.
Our numerical experiments show that
our DC based algorithm can retrieve solutions when $m \approx 2n$.

Furthermore, we develop an $\ell_1$ DC based algorithm to reduce the number of measurements and recover sparse
signals. That is, starting from $\bfx^{(k)}$, we solve
\begin{equation}
\label{L1DCalg}
\bfx^{(k+1)}:= \hbox{arg} \min \lambda \|\bfx\|_1+ F_1(\bfx) - \nabla F_2(\bfx)^\top (\bfx- \bfx^{(k)})
\end{equation}
using a proximal gradient method, where $\lambda>0$ is a parameter.
The convergence of the $\ell_1$ DC based algorithm can be established
similar to  the DC base algorithm. To accelerate the convergence of the $\ell_1$ DC based algorithm, we
use Attouch-Peypouquet's acceleration (cf. \cite{AP16}). To have a better initialization, we use the projection
technique (cf. \cite{FD01}).  In addition, the hard thresholding operator is used to project each iteration
onto the set of sparse vectors. With these updates, the algorithm works very well.
The numerical experiments of the modified $\ell_1$ DC based algorithm can
recover sparse signals as long as $m\approx n$.

\subsection{Organization}
The paper is organized as follows. First, using tools of algebraic geometry, we explain some existence
for phase retrieval and give an estimate of how many distinct solutions in \cref{sec:existence}.
In \cref{sec:phase}, we give the analysis of convergence for our DC based algorithm.
Accelerated gradient descent methods, Nesterov's accelerated technique and the BB technqie for inner
iterations will be discussed in \cref{sec:Nes}. Furthermore, we will study the $\ell_1$ based
algorithm for retrieving sparse signals and discuss the convergence in \cref{sec:sparse}.
Our numerical experimental results are collected in \cref{sec:numerical}, where we show the performance of our DC
based algorithms and
comparison with the Gauss-Newton algorithm for general signals and sparse signals.
Mainly, we will show that the DC based algorithm
is able to retrieve signals when $m\approx 2n$ with high frequency of successes. In addition, our
$\ell_1$ DC based algorithm with the update techniques
is able to retrieve sparse signals with high frequency of successes when $m\approx n$.

\section{On Existence and Number of Phase Retrieval Solutions}
\label{sec:existence}
In this section, we shall discuss the existence of phase retrieval solution and give an estimate
of the number of distinct solutions.
To do so, we first recall  PhaseLift (cf. \cite{CESV12}) which shows the connection between phase retrieval
and low-rank matrix recovery.
%Another approach to recast the phase retrieval problem is called PhaseLift.

Letting $X = \bfx \bfx^\top$ and $A_j = \bfa_j \bfa_j^\top$, $j=1, \ldots, m$,
the constrains in \cref{quadcond} can be rewritten as
\begin{equation*} % \label{quadcond2}
  b_j=\hbox{tr}(A_j X),\quad j=1, \ldots, m,
\end{equation*}
where $\hbox{tr}(\cdot)$ is the trace operator.

Note that a scaling of $\bfx$ by a unimodular constant $c$ would not change $X$. Indeed, $(c\bfx)(c\bfx)^{\top}
= |c|^2 \bfx \bfx^{\top} = \bfx \bfx^{\top} = X$.
Conversely, given a positive semi definite matrix $X$ of rank 1, there exists a vector $\bfx$ such that $X =
\bfx \bfx^{\top}$.
So the phase retrieval problem can be recast as a matrix recover problem (cf. \cite{CESV12}):
Find $ X \in {\cal M}_1$ satisfying linear measurements: $ \hbox{tr}(A_jX)= b_j, j=1, \ldots, m$, where
${\cal M}_r= \{X\in \mathbb{R}^{n\times n} : \hbox{rank}(X) = r\}$.
It also can be considered as a low rank matrix recovery problem:
\begin{equation}
\label{phaseretrieval}
\min \{\hbox{rank}(X):  \hbox{tr}(A_j X)=b_j, j=1, \ldots, m, X\succeq 0\}.
\end{equation}
As we have pointed out above, for any given $b_j\geq 0$, $j=1, \ldots, m$,
there may not have a matrix $X\in {\cal M}_r$ with $r<n$ satisfying the constraints exactly.
Unless $b_j\geq 0$ are exactly the measurement values from a matrix $X$ so that we can use the
minimization \cref{phaseretrieval} to find the solution $X$,
we have to reformulate the above problem otherwise:
\begin{equation}
\label{phaseretrieval2}
\min \{ \sum_{i=1}^m |\hbox{tr}(A_j X)-b_j|^2, X\in {\cal M}_r, X\ge 0\}.
\end{equation}
As ${\cal M}_r$ is a closed set, the above least squares problem will have a bounded solution
if the following coercive condition holds:
\begin{equation*}
%\label{coercivity}
\sum_{i=1}^m |\hbox{tr}(A_j X)-b_j|^2 \to \infty \hbox{ when } \|X\|_F\to \infty.
\end{equation*}
In the case that the above coercive condition does not hold, one has to use other conditions to
insure that the minimizer in \cref{phaseretrieval2} is bounded. For example, if there is a matrix $X_0$
which is orthogonal to $A_j$ in the sense that
$\hbox{tr}(A_j X_0)=0$ for all $j=1, \ldots, m$, then the coercive condition will not hold as one can let
$X= \ell X_0$ with $\ell\to \infty$.

We are now ready to discuss the existence of the solution of phase retrieval problem.
Let ${\cal M}_r$ be the set of matrices of size $n\times n$ with rank $r$ and $\overline{\mathcal{M}}_{r}$
be the set of all matrices with rank $\le r$.
It is known that dimension of ${\cal M}_r$ is $2 n r-r^2$ (cf. Proposition 12.2 in \cite{JoeHarris}
for a proof). Since $\overline{\mathcal{M}_{r}}$ is the closure of ${\cal M}_r$
in the Zariski sense (cf. \cite{ZARISKI}) and hence  the dimension of
$\overline{\mathcal{M}}_r$ is also $2nr- r^2$.
Furthermore,  it is clear that $\overline{\mathcal{M}_{r}}$ is an algebraic variety.
In fact, $\overline{\mathcal{M}_{r}}$  is an irreducible variety which is a standard result in algebraic geometry. To make the paper self-contain,  we present a short proof.
\begin{lemma}\label{irreduciblity}
$\overline{\mathcal{M}_{r}}$ is an irreducible variety.
\end{lemma}
\begin{proof}
Denote by $GL(n)$ the set of invertible $n\times n$ matrices. Consider the action of $GL(n)\times GL(n)$
on $M_n(R)$ given by: $ (G_1,G_2)\cdot M \mapsto G_1 M G_2^{-1}$, for all $G_1, G_2\in GL(n)$.  Fix a rank $r$ matrix $M$.
Then the variety $\mathcal{M}_{r}$ is the orbit of $M$. Hence, we have a surjective morphism,  a regular algebraic map
described by polynomials,  from $GL(n)\times GL(n)$ onto $\mathcal{M}_{r}$.
Since $GL(n)\times GL(n)$ is an irreducible variety, so is $\mathcal{M}_{r}$.
Hence, the closure $\overline{\mathcal{M}_{r_g}}$ of the irreducible set $\mathcal{M}_{r_g}$
is also irreducible \emph{c.f} (cf. Example I.1.4 in \cite{Hartshorne}).
\end{proof}

Define a map
$$
{\cal A} : \mathcal{M}_1 \rightarrow \mathbb{R}^{m}
$$
by projecting any matrix $X \in \mathcal{M}_1$ to $(b_1, \cdots, b_m)^\top \in  \mathbb{R}^{m}$ in the sense that
$${\cal A}(X)= (\hbox{tr}(A_1 X), \cdots, \hbox{tr}(A_m X))^\top.$$
We define the range ${\cal R}_+=\{{\cal A}(X):  X\in {\cal M}_1, X\succeq 0\}$ and the range ${\cal R}
=\{{\cal A}(X):  X\in {\cal M}_1\}$ of the map ${\cal A}$. It is clear that the dimension of ${\cal R}_+$ is less than or
equal to the dimension of ${\cal R}$. As the projection ${\cal A}$ is a regular map since each coordinate of the map ${\cal A}$
is a linear polynomial in entries of matrices,
we expect that $\dim({\cal R})$ is less than or equal to the dimension of the  ${\cal M}_1$
which is equal to $2n-1$.   If $m>2n-1$, then ${\cal R}$ is not able to occupy the whole space $\mathbb{R}^m$.
The Lebesgue measure of the range ${\cal R}$
is zero and hence, randomly choosing a vector $\bfb= (b_1, \cdots, b_m)^\top \in \mathbb{R}^m$, e.g. $\bfb\in \mathbb{R}^m_+$
 will not be in  ${\cal R}$ most likely and hence, not in ${\cal R}_+$. Thus, there will not be a solution $X\in {\cal M}_1$
 such that ${\cal A}(X)=\bfb$.

Certainly, these intuitions should be made more precise.
Recall the following result from Theorem 1.25 in Sec 6.3 of \cite{Shafarevich}.
\begin{lemma}\label{fiberdimensionlemma}
Let $f : X \rightarrow Y$ be a regular map between irreducible varieties.
Suppose that $f$ is surjective: $f(X) = Y$ , and that $\dim(X) = n$, $\dim(Y) = m$. Then $m \leq n$,
and
\begin{enumerate}
	\item for any $y \in Y$ and for any component $F$ of the fiber $f^{-1}(y)$,  $\dim(F) \geq n-m$;
	\item there exists a nonempty open subset $U \subset Y$ such that $\dim(f^{-1}(y)) = n - m$ for $y \in U$.
\end{enumerate}
\end{lemma}
%\begin{proof}
%Refer to Theorem 1.25 in Sec 6.3 of \citet{Shafarevich}.
%\end{proof}

We are now ready to prove
\begin{theorem}
\label{nochance}
If one chooses randomly a vector $\bfb=(b_1, \cdots, b_m)^\top \in \mathbb{R}_+^m$ with $m>2n-1$, the
probability of finding a solution $X$ to the minimization \cref{phaseretrieval} is zero.
In other words, for almost all the $\bfb=(b_1, \cdots, b_m)^\top \in \mathbb{R}_+^m$
the solution to \cref{phaseretrieval}  is a matrix with rank more than or equal to $2$.
\end{theorem}
\begin{proof}
We mainly use \cref{fiberdimensionlemma}. Let $X=\overline{\mathcal{M}_1}$ which is an
irreducible variety  by Lemma~\ref{irreduciblity}.
Let $Y=\{{\cal A}(M), M\in \overline{\mathcal{M}_1}\}$, i.e.
$Y={\cal R}$ which is also an irreducible variety as it is a
continuous image of the irreducible variety $\overline{\mathcal{M}_1}$.
Since ${\cal A}$ is a regular map,  we have $\dim ({\cal R}) \le
\dim(\overline{\mathcal{M}_{1}})=2n-1<m$. Thus,
${\cal R}$ is a proper lower dimensional closed subset in $\mathbb{R}^m$.
For almost all points in $\mathbb{R}^m$, they do not belong to ${\cal R}$.
In other words, for almost all points  $\bfb=(b_1, \cdots, b_n)\in \mathbb{R}^m$, there is no matrix
$M\in \overline{\mathcal{M}_{1}}$ such that  ${\cal A}(M)={\bf b}$ and hence, no matrix $M \in \overline{\mathcal{M}_1}$ with
$M\ge 0$ such that ${\cal A}(M)={\bf b}$ as the set ${\cal R}_+$ is a subset of ${\cal R}$.
\end{proof}

Note that the above discussion is still valid after replacing ${\cal M}_1$ by ${\cal M}_r$ with $r<n$. Under
the assumption
that $m>2nr-r^2$, we can show that the generalized phase retrieval problem \cref{phaseretrieval}
may not have a solution for randomly chosen
$\bfb=(b_1, \cdots, b_m)\in \mathbb{R}^m$.

Next define the subset $\chi_\bfb \subset \overline{\mathcal{M}_{1}}$ by
$$
\chi_\bfb = \left\{M \in   \overline{\mathcal{M}_{1}}
\mid {\cal A}(M) = \bfb \text{ and } {\cal A}^{-1}({\cal A}(M)) \text{ is zero dimensional} \right\}.
$$

As we are working over Noetherian fields like $\mathbb{R}$ or $\mathbb{C}$, it is worthwhile to keep in mind
that all zero dimensional varieties over such fields will have only finitely many points.
%\textcolor{red}{Abe,
%please add your explanation as we use Hermitian matrix in the complex %variable setting.}
Let us consider those
$\bfb\in \mathbb{R}_+^m$ such that the set $\chi_\bfb$ is nonempty.
We are interested in an upper bound on number of solutions one can
find via the minimization \cref{phaseretrieval} when $\chi_\bfb \neq \emptyset$ (\cref{numberofcompletion}).
To do so,  we need more results from algebraic geometry. % from  Proposition 11.12 in \citet{JoeHarris}.

\begin{lemma}[\cite{JoeHarris} Proposition 11.12.]\label{JoeHarrislemma}
Let $X$ be a quasi-projective variety and $\pi: X \rightarrow \mathbb{R}^m$ a regular map; let $Y$ be
closure of the image. For any $p \in X$, let $X_p = \pi^{-1}\pi(p)) \subseteq X$ be the fiber of $\pi$	
through $p$, and let $\mu(p) = \dim_p(X_p)$ be the local dimension of $X_p$ at $p$. Then $\mu(p)$ is an
upper-semi-continuous function of $p$, in the Zariski topology on $X$, i.e. for any $m$ the locus of
points $p \in X$ such that $\dim_p(X_p) > m$ is closed in $X$. Moreover, if $X_0 \subseteq X$
is any irreducible component, $Y_0 \subseteq Y$ the closure of its image and $\mu$ the minimum value
of $\mu(p)$ on $X_0$, then
\begin{equation}
\label{dimension}
\dim(X_0) = \dim(Y_0) + \mu.
\end{equation}
\end{lemma}

As we saw that $\dim({\cal R})\le \dim(\overline{\mathcal{M}_{r}})$,
we can be more precise about these dimensions as shown in the following
\begin{lemma}
\label{dimensions_been_equal}
Assume $m  > \dim(\overline{\mathcal{M}_{r}})$. Then  $\dim(\overline{\mathcal{M}_{r}})
= \dim({\cal R}) $ if and only if $\chi_\bfb \neq \emptyset$ for some $\bfb\in {\cal R}$.
\end{lemma}
\begin{proof} Assume $\dim(\overline{\mathcal{M}_{r}}) = \dim({\cal R})$. Then
using \cref{fiberdimensionlemma}, there exists a nonempty open subset $U \subset
{\cal R})$ such that $\dim({\cal A}^{-1}(\bfb)) = 0$ for all $\bfb \in U$.
This implies that $\chi_\bfb$ has  finitely many points. Hence $\chi_\bfb \neq \emptyset$.
	 	
We now prove the converse. Assume $\chi_\bfb \neq \emptyset$.
We will apply Lemma~\ref{JoeHarrislemma} above by setting $X = \overline{\mathcal{M}_{r}}$,
$Y={\cal A}(\overline{\mathcal{M}_{r}})$ and $\pi = {\cal A}$. (As we apply lemma , please note that it does
not matter whether we take the closure in $\mathbb{P}^{m}$ or in $\mathbb{C}^{m}$ since $\mathbb{C}^{m}$ is an
open set in $\mathbb{P}^{m}$ and the Zariski topology of the affine space $\mathbb{C}^{m}$ is induced from the
Zariski topology of $\mathbb{P}^{m}$. $\overline{\mathcal{M}_{r}}$ is an affine variety. In particular, it is a
quasi-projective variety.)

By our assumption, $\chi_\bfb$ is not empty. It follows that there is a point $p \in Y$
such that $\pi^{-1}(p)$ is zero dimensional. Since zero is the least dimension possible,
we have $\mu = 0$. Hence, using \cref{dimension} above, we have $\dim(\overline{\mathcal{M}_{1}}) =
\dim({\cal R})$.
But dimension does not change upon taking closure. So, $\dim({\cal R}) =\dim(\overline{\cal R})$.
\end{proof}

Finally, we need the following
\begin{definition}
The \emph{degree} of an affine or projective variety of dimension $k$ is the number of intersection points of the variety with
$k$ hyperplanes in general position.
\end{definition}

For example, the degree of the algebraic variety $\overline{\mathcal{M}_r}$ is known.
See Example 14.4.11 in \cite{Fulton}, i.e.
\begin{example}\label{degreelemma}
Degree of the algebraic variety $\overline{\mathcal{M}_r}$ is
$$
\prod_{i=0}^{n-r-1}\frac{\binom{n+i}{r}}{\binom{r+i}{r}}
$$
In particular, the degree of ${\cal M}_1$ is
$$
\prod_{i=0}^{n-2}\frac{n+i}{1+i}.
$$
\end{example}

We are now ready to prove another main result in this section.
\begin{theorem}
\label{numberofcompletion}
Assume that a given vector $\bfb\in \mathbb{R}^m_+$ lies in the range ${\cal R}_+$. Further assume that $\chi_\bfb \neq \emptyset$.  Then, the number of distinct solutions in $\chi_\bfb$ will be less than or
equal to $\displaystyle \prod_{i=0}^{n-2}\frac{n+i}{1+i}$.
\end{theorem}
\begin{proof}
When we fix $m$ entries in $\bfb$, the set of matrices $M$ of rank $1$ such that ${\cal A}(M)=\bfb$ are
exactly the intersection points of the variety $\overline{\mathcal{M}_1}$ with
$m$ hyperplanes, namely the hyperplanes defined by equations of the form
$\langle A_i, M\rangle=b_i, i=1, \cdots, m$. Since
$m>\dim(\overline{\mathcal{M}_r}) = 2n-1$, the number of intersection points (matrices of rank 1) would be
less than
degree of $\overline{\mathcal{M}_1}$ generically. So, in particular, the number of positive semidefinite
matrices $M$ of rank 1 such that ${\cal A}(M)=\bfb$ would be no more than the degree of
$\overline{\mathcal{M}_1}$.
Now using the exact formula for the degree from \cref{degreelemma}, the result follows.
\end{proof}

\section{The DC based Algorithm for Phase Retrieval}
\label{sec:phase}
Recall that we aim to recover $\bfx$ by minimizing
\begin{equation}\label{min F(x)}
  F(\bfx)= \sum_{i=1}^m (f(\langle {\bfa}_i, \bfx\rangle) -b_i)^2:=F_1(\bfx) - F_2(\bfx),
\end{equation}
where $F_1(\bfx)=\sum_{i=1}^m f^2(\langle {\bfa}_i, \bfx\rangle) + b_i^2$ and $F_2(\bfx)= \sum_{i=1}^m
\left(2 b_i f(\langle {\bfa}_i, \bfx\rangle)\right)$. It is easy to see that the minimization in
\cref{min F(x)} can happen in a bounded region $\mathcal{R}$ since the coercive condition $f(x)\to \infty$
when $x\to \infty$. The DC based computational method is as follows. Start from any iterative solution
$\bfx^{(k)}$, we solve the following convex minimization problem:
\begin{equation}
\label{newmethod}
\bfx^{(k+1)}= \hbox{arg}\min_{\bfx\in \mathbb{R}^n} F_1(\bfx) - \nabla F_2(\bfx^{(k)})^\top (\bfx- \bfx^{(k)})
\end{equation}
for $k\ge 1$, where $\bfx^{(1)}$ is an initial guess which will be discussed how to choose later.
Without of loss generality, we always assume $\bfx^{(1)}$ is located in a bounded region $\mathcal{R}$. \\

Our goal in this section is to show $\bfx^{(k)}, k\ge 1$ converges to a critical point.
Later, we will discuss how to find a global minimization by choosing the initial guess $\bfx^{(1)}$
appropriately. Although it is  standard to solve a convex
minimization problem with differentiable minimizing functional, we have to solve \cref{newmethod}
by using an iterative method.
For example, we can use a gradient descent method with various acceleration techniques such as Nesterov's,
BB's and other techniques or the Newton method.
Hence, there will be two iterative procedures. The iterative procedure for solving \cref{newmethod} is an
inner iteration which will be discussed  in the next section.
In this section, we mainly discuss the outer iteration.

We will instate the following assumptions on the function $F_1$ and $F_2$:
\begin{itemize}
\item [(1)]The gradient function $\nabla F_1$ has Lipschitz constant $L_1$ in bounded region $\mathcal{R}$. That is,
$\normm{\nabla F_1(\bfx)-\nabla F_1(\bfy)}\le L_1 \normm{\bfx-\bfy}$ for all vectors $\bfx,\bfy \in \mathcal{R}$.
\item [(2)] $F_2$ is a strongly convex function with parameter $\ell$ in $\mathcal{R}$. That is,
$F_2(\bfy)\ge F_2(\bfx)+\nabla F_2(\bfx)^\T(\bfy-\bfx)+\frac{\ell}{2} \|\bfy - \bfx\|^2$
for all vectors $\bfx,\bfy \in \mathcal{R}$.
\end{itemize}

Note that $H=\sum_{i=1}^m 2 b_i f^{''}(\bfa_i^\T \bfx) \bfa_i \bfa_i^\T$ is the Hessian matrix of function $F_2=2\sum_{i=1}^m b_if(\bfa_i^\T\bfx)$, where
$f^{''}(x)\ge 0 $ since the convexity of $f$. Then the parameter of strong convexity is given by the
minimal eigenvalue of $H$.
%We will show in the appendix that the Hessian at a global
%minimizer of (\ref{LSmin}) is positive definite for the real value setting and is nonnegative definite for
%the complex value setting.
We first introduce a standard result for DC based algorithm:
\begin{theorem}
\label{mjlai07032018}
Assume $F_2$ is a strongly convex function with parameter $\ell$.  Starting from any initial guess
$\bfx^{(1)}$, let $\bfx^{(k+1)}$ be the solution in (\ref{newmethod}) for all $k\ge 1$. Then
\begin{equation}
F(\bfx^{(k+1)}) \le F(\bfx^{(k)}) -  \frac{\ell}{2} \|\bfx^{(k+1)} - \bfx^{(k)}\|^2, \quad \forall k\ge 1
\end{equation}
and $\nabla F_1(\bfx^{(k+1)}) - \nabla F_2(\bfx^{(k)}) = 0$.
\end{theorem}
\begin{proof}
By the strongly convexity of $F_2$, we have
$$
F_2(\bfx^{(k+1)}) \ge F_2(\bfx^{(k)}) + \nabla F_2(\bfx^{(k)})^\top (\bfx^{(k+1)}
-\bfx^{(k)}) + \frac{\ell}{2} \|\bfx^{(k+1)} - \bfx^{(k)}\|^2.
$$
From (\ref{newmethod}), we see that

\begin{eqnarray*}
	F(\bfx^{(k+1)}) &=& F_1(\bfx^{(k+1)}) - F_2(\bfx^{(k+1)}) \\
	 &\le & F_1(\bfx^{(k+1)}) -
	\nabla F_2(\bfx^{(k)})^\top (\bfx^{(k+1)}-\bfx^{(k)}) - F_2(\bfx^{(k)}) -\frac{\ell}{2} \|\bfx^{(k+1)} - \bfx^{(k)}\|^2 \\
	 &\le & F_1(\bfx^{(k)}) - F_2(\bfx^{(k)})  -\frac{\ell}{2} \|\bfx^{(k+1)} - \bfx^{(k)}\|^2
	 =  F(\bfx^{(k)})-\frac{\ell}{2} \|\bfx^{(k+1)} - \bfx^{(k)}\|^2 .
\end{eqnarray*}

The second property $\nabla F_1(\bfx^{(k+1)})- \nabla F_2(\bfx^{(k)})=0$ follows from the minimization
(\ref{newmethod}) directly.
\end{proof}

Next, we use the Kurdyka-\L ojasiewicz (KL) inequality to establish the convergence rate of $\bfx^{(k)}$. We
refer to \cite{AB09}, \cite{ABRS10}, and \cite{XY13} for using the KL inequality for various minimization
problems. The following is our major theorem in this section.
\begin{theorem}
\label{mjlai07032018c}
Suppose that $F(\bfx)=F_1(\bfx)-F_2(\bfx)$ is a real analytic function. Assume the gradient function $\nabla F_1$ has Lipschitz
constant $L_1>0 $ and $F_2$ is a strongly convex function with parameter $\ell>0$ in bounded region $\mathcal{R}$. Starting from
any initial guess $\bfx^{(1)}$, let $\bfx^{(k+1)}$ be the solution in \cref{newmethod} for all $k\ge 1$. Then $\bfx^{(k)}, k\ge 1$
 converges to a critical point of $F$. Furthermore, if we let $\bfx^*$ be the unique limit, then
\begin{equation}
\label{convergencerate}
\|\bfx^{(k+1)}- \bfx^*\|\le C \tau^k
\end{equation}
for a positive constant $C$ and $\tau\in (0, 1)$.
\end{theorem}

To this end, we need the KL inequality which is central to the global convergence analysis.
\begin{definition} [\L ojasiewicz \cite{L63}]
We say a function $f(\bfx)$ satisfies the Kurdyka-Lojasiewicz (KL) property at point $\mathbf{\bar{x}}$ if there exists $\theta \in [0, 1)$ such that
\begin{equation*}
 |f(\bfx)-f(\mathbf{\bar{x}})|^\theta \le C \|\partial f(\bfx)\|
\end{equation*}
in a neighborhood $B(\mathbf{\bar{x}}, \delta)$ for some $\delta>0$, where $C>0$ is a constant independent of $\bfx$. In other words, there exists $\varphi(s)=cs^{1-\theta}$ with
$\theta \in [0, 1)$ such that the KL inequality holds:
\begin{equation} \label{Lojasiewicz}
  \varphi'(|f(\bfx)-f(\mathbf{\bar{x}})|)\|\partial f(\bfx)\|\ge 1
\end{equation}
for any $\bfx\in B(\mathbf{\bar{x}}, \delta)$ with $f(\bfx)\neq f(\mathbf{\bar{x}})$.
\end{definition}

This property is introduced by Lojasiewicz on the real analytic functions, for which \cref{Lojasiewicz} holds in any critical point with $\theta \in [1/2, 1)$.
Later, many extensions of the above inequality are proposed. Typically, the extension in \cite{K98} for the setting of $o$-minimal structure. Recently, the KL inequality is extended to nonsmooth subanalytic functions. See \cite{AB09,ABRS10,XY13}, for application of the KL inequality for study in the optimization. In our setting,  $\theta=1/2$. We shall include an elementary proof to justify our choice of $\theta=1/2$.

\begin{proposition} \label{elementary}
Suppose that $f: \mathbb{R}^n \mapsto \mathbb{R}$ is a continuously twice differentiable function whose Hessian $H(f)(\bfx)$
is invertible at a critical point $\bfx^*$ of $f$. Then there exists a positive constant $C$, an exponent $\theta=1/2$
and a positive $r$ such that
\begin{equation}
\label{eq2}
	|f(\bfx) - f(\bfx^*)|^{1/2} \le C\|\nabla f(\bfx)\|,	\quad \forall \bfx \in  B(\bfx^*,r),	
\end{equation}
where $B(\bfx^*,r)$ is a ball at $\bfx^*$ with radius $r$.
\end{proposition}
\begin{proof}	Since $f$ is continuously twice differentiable, using Taylor formula for f and noting
$f(\bfx^*) = 0$, we have
$$
|f(\bfx) - f(\bfx^*)| \le  c_1\|\bfx - \bfx^*\|^2,	\quad \forall \bfx \in B(\bfx^*,r)
$$
for some $r > 0$. On the other hand, we have $\|\nabla f(\bfx)\| = \|\nabla f(\bfx)- \nabla f(\bfx^*)\|
\ge c_2\|\bfx - \bfx^*\|$ due
to the fact the Hessian is invertible. Thus, \cref{eq2} follows with $\theta = 1/2$ and $C=\sqrt{c_1}/c_2$.
\end{proof}

The importance of the \L ajosiewicz inequality is  the establishment of the above
inequality in \cref{eq2} when $f$ may not have
an invertible Hessian at the critical point $\bfx^*$.
The proof is based on knowledge from algebraic geometry, mainly the curve
selecting lemma. See \cite{K98} for a more general setting.

Let us recall the geometric description of the landscape function $F(\bfx)$ whose Hessian is
restricted strong convex at the global minimizer (cf. \cite{SQW17}).
In the real variable setting, we can even show that the Hessian is positive definite at a global minimizer.
See the Appendix for a proof.
We are now ready to establish \cref{mjlai07032018c}.

\begin{proof}[Proof of \cref{mjlai07032018c}]
 From \cref{mjlai07032018}, we have
\begin{equation} \label{onekey}
\frac{\ell}{2} \|\bfx^{(k+1)} - \bfx^{(k)}\|^2 \le F(\bfx^{(k)}) - F(\bfx^{(k+1)}).
\end{equation}
That is, $F(\bfx^{(k)}), k\ge 1$ is strictly decreasing sequence. Without loss of generality, we assume
$$\mathcal{R}:=\{\bfx\in \mathbb{R}^n, F(\bfx)\le F(\bfx^{(1)})\}.$$
Then the sequence $\{\bfx^{(k)}\}_{k=1}^\infty \subset \mathcal{R}$ is a bounded sequence. Then there exists a cluster point
$\bfx^*$ and a subsequence $\bfx^{(k_i)}$ such that $\bfx^{(k_i)}\to \bfx^*$. Note that $\{F(\bfx^{(k)})\}_{k=1}^\infty$ is a
bounded monotonic descending sequence. Then $F(\bfx^{(k)})\to F(\bfx^*)$ for all $k\ge 1$. We claim that there exists a a positive
constant $C_1$ such that
\begin{equation}
\label{onekey2}
C_1 \|\bfx^{(k+1)} - \bfx^{(k)}\| \le \sqrt{F(\bfx^{(k)})-F(\bfx^*)} - \sqrt{F(\bfx^{(k+1)}) - F(\bfx^*)}
\end{equation}
holds for all $k\ge k_0$ where $k_0$ is large enough. To establish this claim, we need to use \cref{elementary} which is the
well-known Kurdyka-Lojasiewicz inequality. First, we prove that the condition $\|\nabla F(\bfx^{*})\|= 0$ holds. Indeed, using one
of the properties in \cref{mjlai07032018}, we have
$$
\|\nabla F(\bfx^{(k)})\| = \|\nabla F_1(\bfx^{(k)})- \nabla F_2(\bfx^{(k)})\|
=\|\nabla F_1(\bfx^{(k)})- \nabla F_1(\bfx^{(k+1)})\|\le L_1 \|\bfx^{(k)} -\bfx^{(k+1)}\|.
$$
Combining with \cref{onekey}, it gives that $\|\nabla F(\bfx^{(k_i)})\|\to 0$. By the continuity of gradient function, we have
$\|\nabla F(\bfx^{*})\|= 0$ since $\bfx^{(k_i)}\to \bfx^*$.  Next, consider $g(t) = \sqrt{t}$ which is concave over $[0, 1]$, we
have $g(t)- g(s)\ge g'(t)(t-s)$. Then by the Kurdyka-Lojasiewicz inequality, there exists a positive
constant $c_0>0$ and $\delta>0$ such that
\begin{equation}
\label{newkeyinq}
\|g'(F(\bfx) -F(\bfx^*))\nabla F(\bfx))\|\ge c_0>0
\end{equation}
for all $\bfx$ in the neighborhood $B(\bfx^*,\delta)$ of $\bfx^*$. As
$$ F(\bfx^{(k)})-F(\bfx^*)\to 0, \qquad  k \to \infty, $$
then there is an integer $k_0 $ such that for all $k\ge k_0$ it holds
\begin{equation} \label{delta}
  \max{\left(\sqrt{2/\ell}, L_1/(\ell c_0) \right)}\cdot \sqrt{ F(\bfx^{(k)})-F(\bfx^*)} \le \delta/2.
\end{equation}
 Also, $\bfx^{(k_i)}\to \bfx^*$ as $k_i\to \infty$. Without loss of generality, we may assume that $k_0=1$ and $\bfx^{(1)}\in B(\bfx^*,\delta/2)$.
Let us show that $\bfx^{(k)}, k\ge 1$ will be in the neighborhood $B(\bfx^*,\delta)$. We shall use an induction to do so.
By \cref{delta} we have
$$
\|\bfx^{(2)} -\bfx^*\|\le \|\bfx^{(2)}- \bfx^{(1)}\|+ \|\bfx^{(1)}- \bfx^*\|\le \sqrt{2 (F(\bfx^{(1)})-F(\bfx^*)/\ell} +
\|\bfx^{(1)}- \bfx^*\|\le \delta.
$$
Assume that $\bfx^{(k)}\in B(\bfx^*,\delta)$ for $k\le K$.
Multiplying $g'(F(\bfx^{(k)})- F(\bfx^*))$ to both sides of \cref{onekey}, we have
\begin{align}\label{onekey3}
 \frac{\ell}{2} \|\bfx^{(k+1)} - \bfx^{(k)}\|^2 g'(F(\bfx^{(k)})- F(\bfx^*))  & \le
g'(F(\bfx^{(k)})- F(\bfx^*)) \left(F(\bfx^{(k)}) - F(\bfx^{(k+1)})\right)\cr
& \le  \sqrt{F(\bfx^{(k)})- F(\bfx^*)} - \sqrt{F(\bfx^{(k+1)})- F(\bfx^*)}
\end{align}
by using the concavity of $g$. However,  the K-L inequality \cref{newkeyinq} and \cref{mjlai07032018} gives that
\begin{eqnarray} \label{sk}
  |g'(F(\bfx^{(k)}) -F(\bfx^*))|&\ge & \frac{c_0}{\|\nabla F(\bfx^{(k)})\|}
  = \frac{c_0}{\|\nabla F_1(\bfx^{(k)})- \nabla F_2(\bfx^{(k)})\|} \nonumber\\
   &=& \frac{c_0}{\|\nabla F_1(\bfx^{(k)})- \nabla F_1(\bfx^{(k+1)})\|}
   \ge   \frac{c_0}{L_1 \|\bfx^{(k+1)} -\bfx^{(k)}\|}.
\end{eqnarray}
Putting it in \cref{onekey3} gives that
\begin{equation}\label{main inequality}
  \sqrt{F(\bfx^{(k)})- F(\bfx^*)} - \sqrt{F(\bfx^{(k+1)})- F(\bfx^*)} \ge \frac{\ell c_0}{2 L_1}
\|\bfx^{(k)} -\bfx^{(k+1)}\|
\end{equation}
holds for all $2\le k\le K$. It follows that
\begin{equation*}
 \frac{2 L_1}{\ell c_0}\sqrt{F(\bfx^{(1)})-F(\bfx^*)} \ge \sum_{j=1}^{K} \|\bfx^{(j+1)} -\bfx^{(j)}\|.
\end{equation*}
That is, we have
\begin{eqnarray*}
  \|\bfx^{(K+1)}- \bfx^*\| &\le & \|\bfx^{(K+1)}- \bfx^{(1)}\|+ \|\bfx^{(1)} - \bfx^*\| \\
   &\le & \sum_{j=1}^{K} \|\bfx^{(j+1)} -\bfx^{(j)}\|+ \|\bfx^{(1)} - \bfx^*\| \\
   &\le &  \frac{2 L_1}{\ell c_0}\sqrt{F(\bfx^{(1)})-F(\bfx^*)}+ \|\bfx^{(1)} - \bfx^*\|
    \le  \delta,
\end{eqnarray*}
where the last inequality follows from \cref{delta}. Thus, $\bfx^{(K+1)}\in B(\bfx^*, \delta)$. This shows that all $\bfx^{(k)}$ are in $B(\bfx^*, \delta)$ and inequality \cref{main inequality} holds for all $k$. Hence, we arrive at the claim \cref{onekey2} with $C_1=\ell c_0/(2L_1)$. By summing the inequality in \cref{onekey2} above, it follows
$$
\sum_{k\ge 1} \|\bfx^{(k+1)} - \bfx^{(k)}\| \le \frac{1}{C_1} \sqrt{F(\bfx^{(1)})- F(\bfx^*)}.
$$
That is, $\bfx^{(k)}$ is a Cauchy sequence and hence, it is convergent with $\bfx^{(k)}\to \bfx^\star$. Note that $\nabla F(\bfx^\star)=0$, which implies $\bfx^{(k)}$ converges to a critical point of $F$. \\

Next, we turn to prove the second part. Let $S_k= \sum_{i=k}^\infty \|\bfx^{(i+1)}- \bfx^{(i)}\|$. It follows from \cref{main inequality} that
\begin{eqnarray*}
  C_1 S_k &=& \sum_{i=k}^\infty C_1 \|\bfx^{(i+1)}- \bfx^{(i)}\| \\
   &\le & \sum_{i=k}^\infty (\sqrt{F(\bfx^{(i)})- F(\bfx^*)}
- \sqrt{F(\bfx^{(i+1)})- F(\bfx^*)}) \le \sqrt{F(\bfx^{(k)})- F(\bfx^*)}.
\end{eqnarray*}
Recall from (\ref{sk}) that
\begin{equation*}
  \sqrt{F(\bfx^{(k)})- F(\bfx^*)} \le \frac{L_1}{2c_0} \|\bfx^{(k)}  - \bfx^{(k+1)}\| =C_2(S_k- S_{k+1})
\end{equation*}
where $C_2=L_1/(2c_0)$. Combining the two above inequality that
\begin{equation*}
  S_{k+1}\le \frac{C_2- C_1}{C_2} S_k \le \cdots \le \theta^k S_0
\end{equation*}
for $\tau= (C_2- C_1)/(C_2)$. Since $\|\bfx^{(k)}- \bfx^*\| \le S_k$, we  complete the proof.
\end{proof}

\begin{remark}
We should point out that the conditions on $F, F_1$ and $F_2$ in \cref{mjlai07032018c} are not harsh. Notice that for standard phase retrieval, all these conditions are satisfied, especially when the region $\mathcal{R}$ is sufficiently small and near the global minimization by a technical initialization.
\end{remark}

In summary, two obvious consequences are:
\begin{itemize}
  \item [(1)] For any given initial point $\bfx^{(1)}$, let $D=F(\bfx^{(1)})- F(\bfx^\star)>0$, where
$\bfx^\star$ is one of the global minimizer of \cref{min F(x)}. Then
$$
F(\bfx^{(k)}) - F(\bfx^\star) \le D - \frac{\ell}{2}\sum_{j=1}^{k-1} \|\bfx^{(j+1)}- \bfx^{(j)}\|^2.
$$
That is, $\bfx^{(k)}$ is closer to one of global minimizer than the initial guess point.
  \item [(2)]  As our approach can find a critical point, if a global minimizer $\bfx^\star$ is a local minimizer over a neighborhood $N(\bfx^\star)$ and an initial vector $\bfx^{(1)}$ is in $N(\bfx^\star)$, then our approach finds $\bfx^\star$.
\end{itemize}

\begin{example}
In this example, we consider the standard phase retrieval problem where $f(x)=|x|^2$. Assume the measurements are Gaussian random
vectors, it has been showed that one can use the initialization from \cite{TWF,Gaoxu} to find an excellent initial vector. More
specifically, to recover a vector $x\in \RR^n$ (or $x \in \C^n$), if the number of measurement vectors $\bfa_i, i=1,\ldots,m $ is
$m=O(n)$, then with high probability we have
\begin{equation*}
  \normm{\bfx^{(1)}-\bfx^{*}} \le \delta \normm{\bfx^{*}},
\end{equation*}
where $\bfx^{*}$ is a global minimizer and $\delta$ is a positive constant. Furthermore, in a small neighborhood
$N(\bfx^\star,\delta):=\{x:\normm{\bfx-\bfx^{*}}\le \delta \normm{\bfx^{*}}\}$, the minimizing functional $F(x)$ is strongly
convex. Thus, our algorithm can converge to the global minimizer by using a initialization.
\end{example}

\section{Computation of the Inner Minimization \cref{newmethod}}
\label{sec:Nes}
We now discuss how to compute the minimization in \cref{newmethod}. For convenience, we rewrite the minimization
in the following form
\begin{equation}
\label{min2}
\min_{\bfx\in \mathbb{R}^n} G(\bfx)
\end{equation}
for a differentiable convex function $G(\bfx):=F_1(\bfx)- \langle \nabla F_2(\bfx^{(k)}),\bfx - \bfx^{(k)}\rangle$. The first
approach is to use the gradient descent method:
\begin{equation}
\label{GD}
\bfz^{(j+1)}= \bfz^{(j)}- h \nabla G(\bfz^{(j)})
\end{equation}
for $j\in {\mathbb N} $ with $\bfz^{(1)}= \bfx^{(k)}$, where $h>0$ is a fixed step size or variable step size.
It is well-known that we need to choose $h \approx 1/(2L)$ for the Lipschitz differentiability constant $L$ of
$G(\bfx)$ and then the gradient descent method \cref{GD} will have a linear convergence. It is also
known that  we can choose $h=\nu/L$ for the Lipschitz differentiability constant $L$ of $G(\bfx)$ and
the $\nu$-strong convexity of $G$ and then use the Nestrov acceleration technique as explained
in \cite{N04}. The convergence will be sped up. See the following result.
\begin{lemma} [The Nesterov's Acceleration (\cite{N04})]
\label{NAG}
Let $f:\mathbb{R}^n\to \mathbb{R}$ be a $\nu$-strong convex function and the gradient function has $L$-Lipschitz constant.
Start at an arbitrary initial point $\bfu_1=\bfz_1$, the following Nesterov's accelerated gradient descent
\begin{align}
\label{newproblem2}
\bfz^{j+1}: &= \bfu^{(j)} -  \frac{\nu}{L} \nabla f(\bfu^{(j)}), \cr
\bfu^{(j+1)} & =\bfz^{(j+1)}- q (\bfz^{(j+1)} -\bfz^{(j)})
\end{align}
satisfies
\begin{equation}
\label{linearrate}
f(\bfz^{j+1}) - f(\bfz^*) \le \frac{\nu+ L}{2} \|\bfz^{(1)} - \bfz^*\|^2 \exp( -\frac{j}{\sqrt{L/\nu}}),
\end{equation}
where $\bfz^*$ is the optimal solution and $q=(\sqrt{L/\nu}-1)/(\sqrt{L/\nu}+1)$ is a constant.
\end{lemma}
The significance of the Nestrov acceleration above is to reduce the number of iterations in \cref{GD}
significantly. That is,
for any tolerance $\epsilon$, we need $O(1/\epsilon)$ number of iterations for the gradient descent method
due to the linear convergence, but
$O(1/\sqrt{\epsilon})$ number of iterations if Nesterov's acceleration \cref{newproblem2} is used.

Since $G$ is twice differentiable, we can certainly use the Newton method to solve \cref{newmethod}
because it has quadratic convergence.  However,  we will not pursue it here due to the fact that when
the number of variables of $\bfz$ is large, the Newton method will be extremely slow.
Instead, another method to choose a good $h$ is to use the Barzilai-Borwein(BB) method which is an excellent approach
for a large scale minimization problem (cf. \cite{BB88}). The iteration of the BB method can be described as
\begin{equation} \label{BB}
\bfz^{(j+1)}= \bfz^{(j)} - \beta_j^{-1} \nabla G(\bfz^{(j)}),
\end{equation}
where the  step size
\begin{equation} \label{bb}
  \beta_j  =  (\bfz^{(j)}- \bfz^{(j-1)})^\top ( \nabla G(\bfz^{(j)}) - \nabla G(\bfz^{(j-1)}))/
\|\bfz^{(j)}- \bfz^{(j-1)}\|^2.
\end{equation}
We shall use the following \cref{DCinneralg1} to solve the minimization \cref{newmethod}.
\begin{algorithm}
\caption{The BB Algorithm for the Inner Minimization}
\label{DCinneralg1}
\begin{algorithmic}
\STATE{Let $\bfu^{(1)}= \bfz^{(1)}$ be an initial guess.}
\STATE{For $j \ge 1$, we solve the minimization in \cref{min2} by computing $\beta_j$
according to \cref{bb}. }
\STATE{Update
\begin{align}
\label{newproblem3}
\bfz^{(j+1)}: &= \bfu^{(j)} -  \beta_j^{-1} \nabla G(\bfu^{(j)}) \cr
\bfu^{(j+1)} & =\bfz^{(j+1)}- q (\bfz^{(j+1)} -\bfz^{(j)})
\end{align}}
\STATE{until a maximum number $T$ of iteration is achieved.}
\RETURN $\bfu^T$
\end{algorithmic}
\end{algorithm}

Our computation of inner minimization is described in \cref{DCinneralg1}, which is the combination of BB method with
Nesterov's accelerated gradient descent. The intuition behind it based on the results in \cref{NAG}.
Since BB method has a good performance in numerical experiment, we can hope our \cref{DCinneralg1} has better performance.

There are several modified versions of the BB-method available together their convergence analysis in the literature. See, e.g. \cite{DL02,ZZ17} and the references therein.
A quick literature search shows that the convergence rate is still not established yet for general minimizing functional $F$ to the best of the authors knowledge.
Next we give a necessary and sufficient condition for the algorithm (\ref{BB}) has a better convergence than linear rate.
We say a algorithm is convergent superlinearly if
$$
\sigma_k = \frac{\|\bfx^{(k+1)} - \bfx^*\|}{\|\bfx^{(k)}- \bfx^*\|}  \to 0, \hbox{ when } k\to \infty.
$$
To analyze the convergence of the BB method in our setting, let $\bfs_{k+1}=\bfx^{(k+1)}- \bfx^{(k)}$ and
$\bfy_{k+1}= \nabla G(\bfx^{(k+1)})- \nabla G(\bfx^{(k)})$.

\begin{lemma}
\label{lai07092018}
Suppose that the function $G(\bfx)$ in \cref{min2} is $\alpha$-strongly convex and the gradient function has Lipschitz constant
$L$ in a domain $D$. Assume $\bfx^*\in D$ and the sequence $\{\bfx^{(k)}, k\ge 1\}$ obtained from the BB method above remain in
$D$.  Then $\{\bfx^{(k)}, k\ge 1\}$ converges super linearly to $\bfx^*$ if and only if $(\beta_k-
H_G(\bfx^*))\bfs_{k+1}=o(\|\bfs_{k+1}\|)$.
\end{lemma}
\begin{proof}
From iteration \cref{BB}, We have
\begin{align}
\left(\beta_k - H_G(\bfx^*)\right)\bfs_{k+1}
=&  -\nabla G(\bfx^{(k)}) -  H_G(\bfx^* )\bfs_{k+1} \cr
=& \nabla G(\bfx^{(k+1)})-\nabla G(\bfx^{(k)})-H_G(\bfx^*)\bfs_{k+1} - \nabla G(\bfx^{(k+1)}).
\end{align}
Since the Hessian $H_G(\bfx)$ is continuous at $\bfx^*$ and all $\bfx^{(k)}\in D$,  we see
$$
\nabla G(\bfx^{(k+1)})-\nabla G(\bfx^{(k)})-H_G(\bfx^*) \bfs_{k+1} \to 0, \quad k\to \infty.
$$
By the assumption that $(\beta_k-  H_G(\bfx^*))\bfs_{k+1} =o(\|\bfs_{k+1}\|)$, we have
\begin{equation}
\label{property0}
\lim_{k\to \infty} {||\nabla G(\bfx^{(k+1)})||\over || \bfs_{k+1}||} = 0.
\end{equation}
Note that
$$\|\nabla G(\bfx^{(k+1)}) -  G(\bfx^{(k)}) \|\le L \|\bfx^{(k+1)}- \bfx^{(k)}\|$$
and
$$
|| \nabla G(\bfx^{(k+1)})|| = ||\nabla G(\bfx^{(k+1)}) - \nabla G(\bfx^*)||
= ||H_G(\xi_k)(\bfx^{(k+1)}-\bfx^{*})|| \ge \alpha \| \bfx^{(k+1)}-\bfx^{*}||
$$
for $\bfx^{(k+1)}\in D$, where  $\xi_k$ in $D$. Then, we have
$$
 {||\nabla G(\bfx^{(k+1)})||\over ||\bfy_{k+1}||} \ge
{\alpha ||\bfx^{(k+1)}-\bfx^{*}||\over L||\bfx^{(k+1)}-\bfx^{*}|| + L|| \bfx^{(k)}-\bfx^{*}||} =
 {\alpha\sigma_k\over L(1 + \sigma_k)},
$$
where $\displaystyle \sigma_k ={||\bfx^{(k+1)}-\bfx^{*}||\over ||\bfx^{(k)}-\bfx^{*}||}$.  It follows that
$\displaystyle {\sigma_k\over 1 + \sigma_k} \rightarrow 0$ and hence,
$\displaystyle \sigma_k \rightarrow 0$.  That is, the BB method converges super-linearly.

On the other hand, if $\sigma_k\to 0$, we can show that $(\beta_k - H_G(\bfx^*))\bfs_{k+1} = o(\|\bfs_{k+1}\|)$.
In fact,  it is known that when $\bfx^{(k)}\to \bfx^{*}$ super-linearly, then
\begin{equation} \label{property1}
 \lim\limits_{ k \rightarrow +\infty} \ \frac{\|\bfx^{(k+1)}-\bfx^{(k)}||} {||\bfx^{(k)}-\bfx^{*}||} = 1.
\end{equation}
Indeed, we have
$$
\displaystyle \Big| ||\bfx^{(k+1)}-\bfx^{(k)}|| - ||\bfx^{(k)}-\bfx^{*}|| \Big| \le ||\bfx^{(k+1)}-\bfx^{*}||.
$$
It follows that
\begin{equation*}
\left|  {||\bfx^{(k+1)}-\bfx^{(k)}||\over ||\bfx^{(k)}-\bfx^{*}||} - 1 \right| \le
{||\bfx^{(k+1)}-\bfx^{*}||\over ||\bfx^{(k)}-\bfx^{*}||} \rightarrow 0.
\end{equation*}
Hence, we have
\begin{eqnarray*}
  {||\nabla G(\bfx^{(k+1)})||\over || \bfs_{k+1}||}&\le & \frac{||\nabla G(\bfx^{(k+1)})- \nabla G(\bfx^*)||}{|| \bfs_{k+1}||} \\
  &\le &  \frac{ L \|\bfx^{(k+1)} -\bfx^*\|}{\|\bfx^{(k+1)}- \bfx^{(k)}\|} \\
   &=& \frac{\sigma_{k+1}}{\|\bfx^{(k+1)}-\bfx^{(k)}\|/\|\bfx^{(k)}-\bfx^*\|} \to 0
\end{eqnarray*}
because of the denominator is bounded by the property \cref{property1}. Using the argument at the beginning of the proof, we
can see that $(\beta_k - H_G(\bfx^*))\bfs_{k+1}= o(\|\bfs_{k+1}\|)$. These completes the proof.
\end{proof}

\section{Sparse Phase Retrieval}
\label{sec:sparse}
In previous sections, several computational algorithms have been developed
 for the phase retrieval problem based on measurements in \cref{quadcond}.
%was studied by using Wirthinger algorithm (similar to the gradient descend algorithm).
We now extend the approaches to study the sparse phase retrieval. Suppose that $\bfx_\bfb$ is a sparse
solution to the given measurements \cref{quadcond}.  Let us use the DC based
algorithm to explain how to do. First, we consider the following
\begin{equation} \label{mainproblem4SPR}
\min_{\bfx \in \mathbb{R}^n \hbox{ or } \mathbb{C}^n} \lambda \|\bfx\|_1+ \sum_{i=1}^m (f(\langle {\bfa}_i, \bfx\rangle) -b_i)^2
\end{equation}
by adding $\lambda \|\bfx\|_1$ to \cref{LSmin} as a standard approach in compressive sensing.
If we take $f(\langle {\bfa}_i, \bfx\rangle)=\abs{\langle {\bfa}_i, \bfx\rangle}^2$, then \cref{mainproblem4SPR} reduces to the sparse phase retrieval.

We now discuss how to solve it numerically.
We approach it by using a similar method as in the previous section. Indeed, for the case $\bfx\in
\mathbb{R}^n$ and $\bfa_i \in\mathbb{R}^n$,
we rewrite $F(\bfx)= \sum_{i=1}^m (f(\langle {\bfa}_i, \bfx\rangle) -b_i)^2$
as
$$
F(\bfx) = F_1(\bfx) - F_2(\bfx):=
\sum_{i=1}^m f^2(\langle {\bfa}_i, \bfx\rangle)+ b_i^2 - \sum_{i=1}^m 2 b_i f(\langle {\bfa}_i,
\bfx\rangle).
$$
The minimization \cref{mainproblem4SPR} will be approximated by
\begin{equation}
\label{DCalg4SPR}
\bfx^{(k+1)}:= \hbox{arg} \min \lambda \|\bfx\|_1 +  F_1(\bfx) - \nabla F_2(\bfx^{(k)})^\top (\bfx- \bfx^{(k)})
\end{equation}
for any given $\bfx^{(k)}$. We call this algorithm as sparse DC based method.
When $\bfx\in \mathbb{C}^n$ and ${\bfa}_j\in \mathbb{C}^n$,
$j=1, \ldots, m$, we have to write $\bfx= \bfx_R+ \sqrt{-1}\bfx_I$ and similar for ${\bfa}_j$. Letting $\bfc=[\bfx_R^\top\,
\bfx_I^\top]^\top \in \mathbb{R}^{2n}$, we view $F_1(\bfx)$ as  a functions in $G_1(\bfc)
=F_1(\bfx_R+\sqrt{-1}\bfx_I)$.
Then $G_1(\bfc)$ is a convex function of variable $\bfc$. Similarly, $G_2(\bfc)=F_2(\bfx_R+
\sqrt{-1}\bfx_I)$ is a convex function of $\bfc$. We can formulate the same minimization problem
as in \cref{DCalg4SPR}.
For convenience, we simply discuss the case when $\bfx,\;  {\bfa}_j, \; j=1,
\ldots, m$ are real.  The complex variable setting can be treated in the same fashion.

To solve \cref{DCalg4SPR}, we use a proximal gradient method: for any given $\bfy^{(k)}$,
\begin{equation}
\label{DCalg4SPR2}
\bfy^{(k+1)}:= \hbox{arg} \min \lambda \|\bfy\|_1 +  F_1(\bfy^{(k)})
+ (\nabla F_1(\bfy^{(k)})- \nabla F_2(\bfy^{(k)}))^\top (\bfy-\bfy^{(k)}) + \frac{L_1}{2}\|\bfy - \bfy^{(k)}\|^2
\end{equation}
for $k\ge 1$, where $L_1$ is the Lipschitz differentiability of $F_1$.
The above minimization can be easily solved by using shrinkage-thresholding technique as in \cite{BT09}. Note
that Beck and Teboulle in \cite{BT09} use a Nesterov's acceleration technique to speed up the iteration to form
the well-known FISTA.  However,  we shall use the acceleration technique from \cite{AP16} which is
slightly better than the Nestrov technique. The discussion above furnishes a computational method
for sparse phase retrieval problem \cref{mainproblem4SPR}. Let us point out one significant difference
from \cref{DCalg4SPR2} is that one can find $\bfy^{(k+1)}$ by using a formula while the solution $\bfx^{(k+1)}$
in \cref{newmethod} has to be computed using an iterative method as explained before. Thus
the sparse phase retrieval is more efficient in this sense.

Let us study the convergence of our sparse phase retrieval method.
We again start with a standard result for a DC based algorithm:
\begin{theorem}
\label{mjlai08182018}
Assume $F_2$ is a strongly convex function with parameter $\ell$.
Starting from any initial guess $\bfy^{(1)}$, let $\bfy^{(k+1)}$ be the solution in \cref{DCalg4SPR2} for all $k\ge 1$. Then
\begin{equation}
\lambda \|\bfy^{(k+1)}\|_1+ F(\bfy^{(k+1)}) \le \lambda \|\bfy^{(k)}\|_1+
F(\bfy^{(k)}) -  \frac{\ell}{2} \|\bfy^{(k+1)} - \bfy^{(k)}\|^2, \quad \forall k\ge 1
\end{equation}
and $\partial g(\bfy^{(k+1)})+  \nabla F_1(\bfy^{(k)}) - \nabla F_2(\bfy^{(k)})
+\dfrac{L_1}{2}(\bfy^{(k+1)}-\bfy^{(k)}) = 0$,
where $g(\bfx)=\lambda \|\bfx\|_1$ and $\partial g$ stands for the subgradient of $g$.
\end{theorem}
\begin{proof}
The Lipschitz differentiability of $F_1$ tells us
$$
F_1(\bfy^{(k+1)}) \le F_1(\bfy^{(k)}) + \nabla F_2(\bfy^{(k)})^\top (\bfy^{(k+1)}
-\bfy^{(k)}) + \frac{L_1}{2} \|\bfy^{(k+1)} - \bfy^{(k)}\|^2,
$$
where $L_1$ is the Lipschitz differentiability of $F_1$. By the strongly convexity of $F_2$, we have
$$
F_2(\bfy^{(k+1)}) \ge F_2(\bfy^{(k)}) + \nabla F_2(\bfy^{(k)})^\top (\bfy^{(k+1)}
-\bfy^{(k)}) + \frac{\ell}{2} \|\bfy^{(k+1)} - \bfy^{(k)}\|^2.
$$
With the above two inequalities, we see that
\begin{eqnarray*}
&&\lambda \|\bfy^{(k+1)}\|_1+	F(\bfx^{(k+1)}) = \lambda \|\bfy^{(k+1)}\|_1+F_1(\bfy^{(k+1)}) - F_2(\bfy^{(k+1)}) \\
	 &\le &\lambda \|\bfy^{(k+1)}\|_1+ F_1(\bfy^{(k)})+ \nabla F_1(\bfy^{(k)})^\top (\bfy^{(k+1)}-\bfy^{(k)})
+\frac{L}{2}\|\bfy^{(k+1)}- \bfy^{(k)}\|^2 \cr
&&\quad - F_2(\bfy^{(k)})
	-\nabla F_2(\bfy^{(k)})^\top (\bfy^{(k+1)}-\bfy^{(k)})  -\frac{\ell}{2} \|\bfy^{(k+1)} - \bfy^{(k)}\|^2 \\
	&=& F_1(\bfy^{(k)})- F_2(\bfy^{(k)}) -\frac{\ell}{2} \|\bfy^{(k+1)} - \bfy^{(k)}\|^2 \cr
&& + \lambda \|\bfy^{(k+1)}\|_1+ (\nabla F_1(\bfy^{(k)}) -
	\nabla F_2(\bfy^{(k)})^\top (\bfy^{(k+1)}-\bfy^{(k)}) +\frac{L}{2}\|\bfy^{(k+1)}- \bfy^{(k)}\|^2 \cr
	 &\le & F_1(\bfy^{(k)}) - F_2(\bfy^{(k)})  -\frac{\ell}{2} \|\bfy^{(k+1)} - \bfy^{(k)}\|^2
+ \lambda \|\bfy^{(k)}\|_1\\
	 &= & \lambda \|\bfy^{(k)}\|_1+ F(\bfy^{(k)})-\frac{\ell}{2} \|\bfy^{(k+1)} - \bfy^{(k)}\|^2,
\end{eqnarray*}
where we have used the optimization condition in (\ref{DCalg4SPR2}). Letting $g(\bfx)= \lambda \|\bfx\|_1$,
the second property $\partial g(\bfy^{(k+1)})+ \nabla F_1(\bfy^{(k)})- \nabla F_2(\bfy^{(k)})
+\dfrac{L_1}{2}(\bfy^{(k+1)}-\bfy^{(k)})=0$
follows from the minimization \cref{DCalg4SPR2}.
\end{proof}

Next we show that the sequence $\bfy^{(k)}, k\ge 1$ from \cref{DCalg4SPR2} converges to a critical point $\bfy^*$.
\begin{theorem}
\label{mjlai8282018}
Suppose that $f(\bfx)$ is a real analytic function and the gradient function $\nabla f(\bfx) $ has Lipschitz constant $L$.
Let $\bfy^{(k)}, k\ge 1$ be the sequence obtained from \cref{DCalg4SPR2}. Then it converges to a critical point $\bfy^*$ of $F$.
\end{theorem}
\begin{proof}
Recall $F(\bfx)= g(\bfx) + f(\bfx)$.  From Theorem~\ref{mjlai08182018}, we have
\begin{equation}
\label{onekeyS}
\frac{\ell}{2} \|\bfy^{(k+1)} - \bfy^{(k)}\|^2 \le F(\bfy^{(k)}) - F(\bfy^{(k+1)}).
\end{equation}
That is, $F(\bfy^{(k)}), k\ge 1$ is strictly decreasing sequence. Due to the coerciveness,  we know that
$$\mathcal{R}:=\{\bfx\in \mathbb{R}^n, F(\bfy)\le F(\bfy^{(1)})\}$$
is a bounded set.
It follows that the sequence $\{\bfy^{(k)}\}_{k=1}^\infty \subset \mathcal{R}$ is a bounded sequence and
there exists a cluster point $\bfy^*$ and a subsequence $\bfy^{(k_i)}$ such that $\bfy^{(k_i)}\to \bfy^*$.
Note that $\{F(\bfy^{(k)})\}_{k=1}^\infty$ is a bounded monotonic descending sequence, then $F(\bfy^{(k)})\to F(\bfy^*)$ for all $k\ge 1$.
We claim that the sequence $\{\bfy^{(k)}\}_{k=1}^\infty$ has finite length, that is,
\begin{equation}\label{claim finite}
  \sum_{k=1}^\infty \|\bfy^{(k+1)}-\bfy^{(k)}\|< \infty.
\end{equation}
To establish the claim, we need to use the Kurdyka-\L ojasiewicz inequality (cf. \cite{K98}).
Note that the $\ell_1$ norm $\|\bfx\|_1$ is semialgebraic function and the function $f(\bfx)$ is analytic,
so the objective function $F(\bfx)$ satisfies the KL property at any critical point (cf. \cite{AB09},
\cite{ABRS10}, \cite{XY13}).
Let us prove that
$\|\nabla F(\bfy^{*})\|= 0$ holds, that is, $\bfy^*$ is a critical point of $F$.
Indeed, using one of the properties in \cref{mjlai08182018}, we have
\begin{align*}
\|\partial F(\bfy^{(k)})\| = & \|\partial g(\bfy^{(k)})+\nabla F_1(\bfy^{(k)})- \nabla F_2(\bfy^{(k)})\|\cr
\le & \|\nabla F_1(\bfy^{(k)})- \nabla F_1(\bfy^{(k-1)})\| + \|\nabla F_2(\bfy^{(k)})-
\nabla F_2(\bfy^{(k-1)})\| \\
 &+ \frac{L}{2} \|\bfy^{(k)}- \bfy^{(k-1)}\|
\end{align*}
by using the second conclusion of \cref{mjlai08182018}. Combining with \cref{onekeyS} and the Lipschitz differentiation of $F_1$ and $F_2$,
it gives that $\|\partial F(\bfy^{(k_i)})\|\to 0$. By a property of subgradient of $g$ (cf. \cite{R70})
and the continuity of the gradients $F_1$ and $F_2$,
we have $\|\partial F(\bfy^{*})\|= 0$ when $\bfy^{(k_i)}\to \bfy^*$.
Thus, $\bfy^*\in \hbox{domain}(\partial F)$, the set of all critical points of $F$.

Therefore, we can use KL inequality to obtain that
\begin{equation}\label{kl ineq}
  \varphi'(F(\bfy)-F(\bfy^*))\|\partial F(\bfy)\|\ge 1
\end{equation}
for all $\bfy$ in the neighborhood $B(\bfy^*,\delta)$.  As
$ F(\bfy^{(k)})-F(\bfy^*)\to 0, \, k \to \infty, $
there is an integer $k_0 $ such that for all $k\ge k_0$ it holds
\begin{equation} \label{sparse delta}
  \max{\left(\sqrt{2/\ell} \sqrt{ F(\bfy^{(k)})-F(\bfy^*)}, 2C/\ell \cdot \varphi(F(\bfy^{(k)})-F(\bfy^*))  \right)}  \le \delta/2.
\end{equation}
Without loss of generality, we may assume that $k_0=1$ and $\bfy^{(1)}\in B(\bfy^*,\delta/2)$.
Let us show that $\bfy^{(k)}, k\ge 1$ will be in the neighborhood $B(\bfy^*,\delta)$.
We shall use an induction to do so.
By \cref{sparse delta} we have
$$
\|\bfy^{(2)} -\bfy^*\|\le \|\bfy^{(2)}- \bfy^{(1)}\|+ \|\bfy^{(1)}- \bfy^*\|\le \sqrt{2 (F(\bfy^{(1)})-F(\bfy^*)/\ell} +
\|\bfy^{(1)}- \bfy^*\|\le \delta.
$$
Assume that $\bfy^{(k)}\in B(\bfy^*,\delta)$ for $k\le K$. From \cref{mjlai08182018}, we have
\begin{eqnarray*}
  \|\partial F(\bfy^{k+1})\| &=& \|\partial g(\bfy^{k+1})+\nabla f(\bfy^{k+1})\| \\
   &=& \|\nabla f(\bfy^{k+1})-\nabla f(\bfy^{k})- \frac{L_1}{2}(\bfy^{k+1}-\bfy^{k})\| \le
 C \|\bfy^{k+1}-\bfy^{k}\|,
\end{eqnarray*}
where constant $C:=L+L_1/2$. Putting it into \cref{kl ineq}, it gives that
\begin{equation} \label{varphi diff}
  \varphi'(F(\bfy^k)-F(\bfy^*)) \ge \frac{1}{C \|\bfy^{k}-\bfy^{k-1}\|}.
\end{equation}
On the other hand, from the concavity of $\varphi$ we get that
\begin{equation*}
  \varphi(F(\bfy^k)-F(\bfy^*))-\varphi(F(\bfy^{k+1})-F(\bfy^*)) \ge \varphi'(F(\bfy^k)-F(\bfy^*))(F(\bfy^{k})-F(\bfy^{k+1})).
\end{equation*}
Combining with \cref{onekeyS} and \cref{varphi diff}, we obtain
\begin{equation*}
  \varphi(F(\bfy^k)-F(\bfy^*))-\varphi(F(\bfy^{k+1})-F(\bfy^*)) \ge \frac{\ell}{2C}\cdot\frac{\|\bfy^{k+1}-\bfy^{k}\|^2}{\|\bfy^{k}-\bfy^{k-1}\|}.
\end{equation*}
Multiplying $\|\bfy^{(k)} -\bfy^{(k-1)}\|$ both sides of the above inequality, 	taking a square root both sides, and
the using a standard inequality $2ab\le a^2+ b^2$ on the left-hand side, we have
$$
\|\bfy^{(k)} -\bfy^{(k-1)}\| + \frac{2C}{\ell} (\varphi(F(\bfy^k)-F(\bfy^*))-\varphi(F(\bfy^{k+1})-F(\bfy^*)))
\ge 2 \|\bfy^{(k)} -\bfy^{(k+1)}\|
$$
for all $2\le k\le K$. It follows that
\begin{equation}\label{spase main inequality}
  \frac{2C}{\ell}\varphi(F(\bfy^{(1)})-F(\bfy^*)) \ge \sum_{j=1}^{K} \|\bfy^{(j+1)} -\bfy^{(j)}\| + \|\bfy^{(K+1)} -\bfy^{(K)}\|.
\end{equation}
That is, we have
\begin{align*}
\|\bfy^{(K+1)}- \bfy^*\| &\le \|\bfy^{(K+1)}- \bfy^{(1)}\|+ \|\bfy^{(1)} - \bfy^*\|
\le \sum_{j=1}^{K} \|\bfy^{(j+1)} -\bfy^{(j)}\|+ \|\bfy^{(1)} - \bfy^*\| \cr
& \le \frac{2C}{\ell}\varphi(F(\bfy^{(1)})-F(\bfy^*))+ \|\bfy^{(1)} - \bfy^*\|\le \delta.
\end{align*}
That is, $\bfy^{(K+1)}\in B(\bfy^*, \delta)$, which implies that all $\bfy^{(k)}$ are in $B(\bfy^*, \delta)$. From above, we know that the inequality \cref{spase main inequality} holds for all $k$, which show the claim \cref{claim finite} holds. It is clear that \cref{claim finite} implies that $\{\bfy^{(k)}\}_{k=1}^\infty$ is a Cauchy sequence and hence, it is convergent with $\bfy^{(k)}\to \bfy^\star$. Note that $\nabla F(\bfy^\star)=0$, which implies $\bfy^{(k)}$ converges to a critical point of $F$.
\end{proof}

Finally, let us show that the convergence is in a linear fashion. We begin with
\begin{lemma}
\label{mjlai08312018}
Let $g(\bfx) =\lambda \|\bfx\|_1$ for $\lambda>0$. Then for any $\bfx$, there exists a $\delta>0$ such that for any $\bfy\in B(\bfx, \delta)$, the open ball of radius $\delta$ at $\bfx$, there exists a subgradient
$\nabla g$ at $\bfx$,
\begin{equation}
\label{mjlai08312018a}
(\partial g(\bfy) - \partial g(\bfx))^\top (\bfy- \bfx)= 0.
\end{equation}
\end{lemma}
\begin{proof}
For similicity, consider $\bfx\in \mathbb{R}^1$. Then if $\bfx\not=0$, we can find $\delta=|\bfx|>0$ such that
when $\bfy\in B(\bfx, \delta)$, we have $\partial g(\bfy)= \partial g(\bfx)$ and hence, we have \cref{mjlai08312018a}.
If $\bfx=0$, for any $y\not=0$, we choose $\partial g(0)$ according to $\bfy$, i.e. $\partial g(0)=1$ if $\bfy>0$
and $\partial g(0)=-1$ if $\bfy<0$. Then we have \cref{mjlai08312018a}.
\end{proof}

In the following lemma, we need to use the sparse set ${\cal R}_s$
\begin{equation}
\label{collection}
{\cal R}_s=\{\bfx\in \mathbb{R}^n|  \|\bfx\|_0\le s\}  = \bigcup_{I \subset \{1,\cdots, n\}\atop |I|=s}
\mathbb{R}^s_I
\end{equation}
which is clearly the union of all canonical subspaces $\mathbb{R}^s_I=\hbox{span}\{\bfe_{i_1}, \cdots,
\bfe_{i_s}\}$ if $I=\{i_1, i_2, \cdots, i_s\}$.

\begin{lemma}
\label{mjlai08312018c}
Let $F(\bfx)= g(\bfx)+ f(\bfx)$. Suppose that $f$ is L-Lipschitz differentiable.
Let $\bfx^*$ be a critical point of
$F$ as explained in \cref{mjlai8282018}. Suppose that either none of entries of $\bfx^*$ is zero or
suppose $\bfx\in \mathbb{R}^s_I$ if $\bfx^*\in \mathbb{R}^s_I$ for some $s\in \{1, \cdots, n\}$.
Then there exists $\delta>0$, e.g. $\displaystyle \delta=\min_{|x^*_i|\not=0} |x^*_i|$  such that for all
$\bfx\in B(\bfx^*, \delta)$,
\begin{equation}
\label{L1newkey}
|F(\bfx)- F(\bfx^*)|\le C\|\bfx  -\bfx^*\|^2.
\end{equation}
\end{lemma}
\begin{proof}
At $\bfx^*$, we have $\partial F(\bfx^*) = \partial g(\bfx^*)+ \nabla f(\bfx^*)=0$.
By using \cref{mjlai08312018} and either one of the assumptions, we have
\begin{align*}
F(\bfx)- F(\bfx^*) &=  g(\bfx)- g(\bfx^*) + f(\bfx)- f(\bfx^*) \cr
&\le \partial g(\bfx)^\top (\bfx-\bfx^*) +
\nabla f(\bfx*)(\bfx- \bfx^*) + \frac{1}{2} (\bfx-\bfx^*)^\top \nabla^2f(\xi) (\bfx- \bfx^*)\cr
&= (\partial g(\bfx)- \partial g(\bfx^*))^\top (\bfx-\bfx^*) + \frac{1}{2} (\bfx-\bfx^*)^\top \nabla^2f(\xi) (\bfx- \bfx^*)\cr
&= \frac{1}{2} (\bfx-\bfx^*)^\top \nabla^2f(\xi) (\bfx- \bfx^*),
\end{align*}
where $\xi$ is a point in between $\bfx^*$ and $\bfx$.
That is, $|F(\bfx)- F(\bfx^*) |\le C \|\bfx -\bfx^*\|^2$ for a positive constant $C$.
\end{proof}

We are now ready to establish the following result on the rate of convergence
\begin{theorem}
\label{mjlai08312018b}
Suppose that $F_2$ is strongly convex.
Starting from any initial guess $\bfx^{(1)}$, let $\bfx^{(k+1)}$ be the solution in \cref{newmethod}
for all $k\ge 1$. Without loss of generality, we assume that
  $\bfx^{(k)}, k\ge 1$ converge to a critical point $\bfx^*$ of $F$ by using \cref{mjlai8282018}.
Then for any $\epsilon>0$, either $\bfx^{(k+1)}\in B(\bfx^*, \epsilon)$ or
\begin{equation}
\label{convergencerate2}
\|\bfx^{(k+1)}- \bfx^*\|\le C_\epsilon \tau^k
\end{equation}
for a positive constant $C_\epsilon$ dependent on $\epsilon$ and $\tau\in (0, 1)$ independent of $\epsilon$.
\end{theorem}
\begin{proof}
According to the results in \cref{mjlai08182018} and \cref{mjlai8282018}, we have
\begin{equation}
\label{L1onekey22}
C_0 \|\bfx^{(k+1)} - \bfx^{(k)}\|^2 \le (F(\bfx^{(k)})-F(\bfx^*)) - (F(\bfx^{(k+1)}) - F(\bfx^*))
\end{equation}
for a positive constant $C_0$.   We now claim that
\begin{equation} \label{Lonekey2}
C_1 \|\bfx^{(k+1)} - \bfx^{(k)}\| \le \sqrt{F(\bfx^{(k)})- F(\bfx^*)} - \sqrt{F(\bfx^{(k+1)}) - F(\bfx^*)}
\end{equation}
for a positive constant $C_1$.
To establish this claim, we need to use the result in \cref{mjlai08312018c}. Let us rewrite the inequality
in \cref{mjlai08312018c} as
$$
\frac{1}{\sqrt{F(\bfx^{(k+1)})- F(\bfx^*)}} \ge \frac{C}{\|\bfx^{(k+1)} -\bfx^*\|}.
$$
Multiplying the above inequality to the inequality in \cref{Lonekey2}, we have
\begin{equation}
\label{L1key3}
C_0 C \frac{\|\bfx^{(k+1)} - \bfx^{(k)}\|^2}{\|\bfx^{(k+1)} -\bfx^*\|}
\le \frac{(F(\bfx^{(k)})-F(\bfx^*)) - (F(\bfx^{(k+1)}) - F(\bfx^*))}{\sqrt{F(\bfx^{(k+1)})- F(\bfx^*)}}
\end{equation}
Consider $g(t) = \sqrt{t}$ which is concave over $[0, 1]$, we know $g(t)- g(s)\ge g'(t)(t-s)$.
Thus, the right-hand side above is less than or equal to the right-hand side of \cref{Lonekey2}.
We now work on the left-hand side of the inequality above. Let us first note that $F$ is strongly convex
outside the ball $B(\bfx_\bfb, \epsilon)$ (in the real variable setting). If $\bfx^{(k+1)}$ is within the
$B(\bfx_\bfb, \epsilon)$, we do not need to do iterations further when $\epsilon>0$ is a tolerance.
Otherwise, we use the strong convexity of $F$ to have
$$
C_\epsilon \|\bfx^{(k+1)}- \bfx^*\|\le \|\nabla F(\bfx^{(k+1)})- \nabla F(\bfx^*)\|
$$
for a positive constant dependent on $\epsilon$. The second property of \cref{mjlai08182018} implies
$$
\partial g(\bfx^{(k+1)}) + \nabla F_1(\bfx^{(k)} - \nabla F_2(\bfx^{(k)}) + \frac{L_1}{2}(\bfx^{(k+1)}
-\bfx^{(k)}) =0
$$
and $\partial g(\bfx^*)+ \nabla f(\bfx^*) =0$. By using \cref{mjlai08312018}, it follows that
$$
\nabla F(\bfx^{(k+1)})- \nabla F(\bfx^*) =
\nabla f(\bfx^{(k+1)})- \nabla f(\bfx^{(k)}) - \frac{L_1}{2}(\bfx^{(k+1)} -\bfx^{(k)}).
$$
In other works,
$$
C_\epsilon \|\bfx^{(k+1)}- \bfx^*\| \le \|\nabla f(\bfx^{(k)})- \nabla f(\bfx^{(k+1)})\| + \frac{L_1}{2}\|\bfx^{(k+1)} -\bfx^{(k)}\|
$$
Using the Lipschitz differentiability of $f$, we have
$$
\|\bfx^{(k+1)}- \bfx^*\| \le \frac{L+L_1}{C_\epsilon} \|\bfx^{(k+1)} -\bfx^{(k)}\|.
$$
The left-hand side of the equation in \cref{L1key3} can be simplified to be
$$
C_0C \frac{L_1+L}{C_\epsilon} \|\bfx^{(k+1)} -\bfx^{(k)}\|
$$
which is the desired term on the left-hand of the inequality in \cref{Lonekey2}. These establish the claim.

By summing the inequality in \cref{Lonekey2} above, it follows
$$
\sum_{k\ge 1} \|\bfx^{(k+1)} - \bfx^{(k)}\| \le \frac{1}{C_1} \sqrt{f(\bfx^{(1)})- f(\bfx^*)}.
$$
That is, $\bfx^{(k)}$ is a Cauchy sequence and hence, it is convergent.

The remaining part of the proof is to establish the convergence rate. The proof is similar to the one in
a previous section. We leave the detail to the interested reader.
\end{proof}

\section{Numerical  Results}
\label{sec:numerical}
In this section, we report some computational results from our DC based algorithm and $\ell_1$ DC based
algorithm.
The significance of these results is  to demonstrate that the DC based algorithm is able
to retrieve real signals of size $n$ from $m$ measurements with high probability around 80\%
over 1000  repeated runs when $m\approx 2n$.  As demonstrated in \cite{TWF} (cf. Figures 8 and 9),
the truncated Wirtinger flow algorithm,
and the original Wirtinger flow algorithm needs $m\approx 3n$ to be able to retrieve Gaussian random signals of size $n$. To retrieve sparse solution, the $\ell_1$ DC based algorithm together with thresholding technique needs only $m\approx n$ measurements.
Thus this section is divided into two subsections. We first present how to use our DC based algorithm to retrieve general signals in the real and complex variable settings.
As the Wirtinger flow (Wf) algorithm requires $m\ge 3n$ to be able to retrieve
the real variable solution, we shall not show the performance of the Wf algorithm.
Instead, we shall also present the Gauss-Newton algorithm from \cite{Gaoxu}
to compare with our DC based algorithm.  Next we shall present numerical
experimental results to demonstrate our $\ell_1$ DC based algorithm is able to use $m\approx n$
measurements to retrieve sparse signals.

\subsection{Phase Retrieval of General Signals}
\begin{example}
\label{prex1}
In this example, we recover the solution $\bfx_\bfb$ from the given measurements \cref{quadcond} using
Gaussian random measurement vectors $\bfa_j, j=1, \cdots, m$. The number $m$ of measurements is around the twice
of the size of the solution $\bfx_\bfb$. In \cref{prtab1}, we show the number of successes of retrieving
$\bfx_\bfb$ over $1000$ repeated runs.  We fix $n=128$ and $m=k*n/16$ for $k=12, 13, \cdots, 35$.

%\centerline{Numerical Results of  the DC based Algorithm}
\begin{table}[htbp]
{\footnotesize
\caption{The numbers of successful retrieved solution over 1000 repeated runs based on numbers of
measurements satisfying the relations $m/n$ listed above, where $n$ is the size of the solution \label{prtab1}}
\begin{center}
\begin{tabular}{|c|ccccccc|}\hline
$m/n$ & 1.3750  &   1.4375 &   1.5000  &  1.5625  &  1.6250 &   1.6875  &  1.7500 \cr \hline
successes & 20 &   48 & 107 & 150 &  239 & 284 & 446\cr \hline \hline
$m/n$ & 1.8125  & 1.8750  &  1.9375 &   2 &   2.0625  & 2.1250  &  2.1875\cr \hline
successes & 511 &  588 &  650&  708 &  771 &  844 &  882 \cr\hline
\end{tabular}
\end{center}
}
\end{table}

From \cref{prtab1}, we can see that the DC based algorithm can retrieve the solutions using the number
$m$ of measurements around $2n$ with high probabilities $\ge 70\%$.
\end{example}

\begin{example}
We next repeat \cref{prex1} using more number of measurements. In this case, we are able to use the
Gauss-Newton method to retrieve solutions as the Hessian can be inverted.
Hence, we will compare the numbers of successes from the
Gauss-Newton method and the DC based method in \cref{prtab2}.
\begin{table}[htbp]
{\footnotesize
\caption{The numbers of successful retrieved solution over 1000 repeated runs based on numbers of measurements satisfying the relations $m/n$ listed above, where $n=128$ is the size of the solution \label{prtab2}}
\begin{center}
\begin{tabular}{|c|ccccccc|}\hline
$m/n$ &  2.3125  &  2.3750 &   2.4375  &  2.5000  &  2.5625 &   2.6250  &  2.6875 \cr \hline
GN successes & 560 &  672 &  708 & 776 &  831 &  882 &  912 \cr \hline
DC successes &  937 & 957 & 961&   967 &    991 &       991 &        989 \cr \hline \hline
$m/n$ & 2.7500  &  2.8125  & 2.8750 &   2.9375  &  3 &  3.0625   & 3.1250\cr \hline
GN successes & 939 & 950 &  960 &  980 &  986 &  987 &  991 \cr\hline
DC successes &  995 & 994 & 995 &  1000 &  1000 &    998 & 998\cr \hline
\end{tabular}
\end{center}
}
\end{table}
\end{example}

\begin{example}
This example shows that robustness of the DC based algorithm. We repeat the computation in \cref{prex1}
by adding noises to the measurements.  One way to generate a noisy input is to add noises $\epsilon_j$ to
the clean measurements:
\begin{equation}
\label{noisym}
\hat{b}_j= |\langle \bfa_j, \bfx_\bfb\rangle|^2+ \epsilon_j, \quad j=1, \cdots, m
\end{equation}
Another way to generate a noisy input is to add noises $\delta_j$ and $\epsilon_j$ to the clean measurements:
\begin{equation}
\label{noisym2}
\tilde{b}_j= |\langle \bfa_j, \bfx_\bfb\rangle+\delta_j|^2+ \epsilon_j, \quad j=1, \cdots, m
\end{equation}
For noisy measurements of model \cref{noisym}, we assume that $\epsilon_j$ are subject to uniform random
distribution between $[-u,u]$ with mean zero, where $u=1e-1$, $1e-3$ and $1e-5$.   As long as the tolerance
as a stopping criterion for the Gauss-Newton method and DC based method is the same as $e$ or bigger than $e$
both algorithms produce the same successes of retrieval as in Table~\ref{prtab2}.
For noisy measurements of model \cref{noisym2}, we assume that both $\epsilon_j$ and $\delta_j$ are subject to
uniform distribution between $[-u,u]$ with mean zero. When the tolerance is the same as $e$ or larger than $e$,
both algorithms can retrieve the solution from noisy measurements \cref{noisym2} just as the same as
in the previous example.
\end{example}

In addition to the retrieve real variable solutions, we shall also repeat the same experiments for complex
variable solutions.
\begin{example}
In this example, we use the DC based algorithm and the
Gauss-Newton method to retrieve complex variable solutions. The number of measurements vs. the number of
entries of complex variable solutions is around 3. That is,  $m\approx 3n$ with $n=128$. For Gaussian random
measurements $\bfa_j= \bfa_{j,R}+\bfi \bfa_{j,I}$, $j=1, \cdots, m$, we retrieve $\bfz\in \mathbb{R}^n$ with
$\bfz =\bfx+ \bfi \bfy$ and $\bfx, \bfy\in \mathbb{R}^n$ from $|\langle  \bfa_j, \bfz\rangle|^2, j=1, \cdots,
m$.   We will show the numbers of successes from the
Gauss-Newton method and the DC based method over 100 repeated runs  in \cref{prtab3}.
\begin{table}[htbp]
{\footnotesize
\caption{The numbers of successful retrieved solution over 100 repeated runs based on numbers of measurements
satisfying the relations $m/n$ listed above, where $n=128$ is the size of the solution \label{prtab3}}
\begin{center}
\begin{tabular}{|c|cccccccccc|}\hline
$m/n$ &   2.562 &  2.625 &   2.687  &  2.750  &   2.812  &  2.875 &   2.937&
3.000  & 3.062  & 3.125\cr \hline
GN alg. &   7 &   21 &   21 &   37  &  56  &   60  &   75 &   72  & 84  &  88\cr  \hline
DC alg. & 0   &   0 &    0 &    3 &   19  &   52 &    78 &   80 &   94  &  99\cr \hline
\end{tabular}
\end{center}
}
\end{table}
The \cref{prtab3} shows that the DC based algorithm is able to retrieve complex variable solutions very well
when $m\ge 3n$ while the Gauss-Newton method can find solutions even when $m \ge 2.5n$. However, the
successful rate is lower than the DC based algorithms starting from $m \ge 2.9n$.
\end{example}

We now present some numerical results to demonstrate that the $\ell_1$ DC based \cref{L1DCalgsparse} works well.

\begin{algorithm}
\caption{$\ell_1$ DC based Algorithm}
\label{L1DCalgsparse}
\begin{algorithmic}
\STATE{We use the same initialization as in the previous examples.}
\WHILE{$k\ge 1$}
\STATE{Solve \cref{DCalg4SPR2} to get $\bfy^{(k+1)}$.}
\STATE{Apply a modified Attouch-Peypouquet technique  to get a new  $\bfy^{(k+1)}$.}
\STATE{Until the maximal number of iterations is reached.}
\ENDWHILE
\RETURN $\bfy^T$
\end{algorithmic}
\end{algorithm}

In \cref{L1DCalgsparse}, the modified Attouch-Peypouquet technique is to
use the Attouch-Peypouquet iteration (cf. \cite{AP16})
in the first few $k$ iterations, say $k\le K$ with variable step size $\beta_k=k/(k+\alpha)$ and then
a fixed step size $\beta_K$ for the remaining iterations.

\begin{example}
We have experimented \cref{L1DCalgsparse} numerically for retrieving solutions in the real variable setting.
We use $n=100$ and $m=1.1n, 1.2n, \ldots, 2.5n$. All measurement vectors $\bfa_j, j=1, \ldots, m$ are Gaussian
random vectors. So is $\bfx_\bfb$. We use \cref{L1DCalgsparse} to recover $\bfx_\bfb$ from
$b_j = |\langle \bfa_j, \bfx_\bfb\rangle|^2, j=1, \ldots, m$ based on $5000$ iterations when $n=100$. To
recover a general solution $\bfx_\bfb$, we use a small value $\lambda=1e-5$. We repeat the experiment 100 times
and summarize the frequency of retrievals listed in \cref{prtab6}.
\begin{table}[htbp]
\caption{The numbers of successful retrieved solution over 100 repeated runs based on numbers of measurements
satisfying the relations $m/n$ listed above, where $n=100$ is the size of the solution \label{prtab6}}
\centering
\begin{tabular}{|c|ccccccccccc|}\hline
$m/n$ &   1.5 &  1.6 &   1.7  &  1.8  &   1.9  &  2 &   2.1& 2.2  & 2.3  & 2.4 & 2.5\cr \hline
$\ell_1$ DC alg. & 0   &  0  &    2  &   8  &   28 &   57 &   72  &   91 &    93 &  93   &  99\cr \hline
\end{tabular}
\end{table}
\end{example}

\subsection{Phase Retrieval of Sparse Signals}
Next we explain how to use our $\ell_1$ DC based algorithm to retrieve sparse solutions. A key point is to use
$m$ measurement values with $m$ smaller than $2n$ or even small than $n$. In such a setting,
many existing algorithms will fail.
As shown in the subsection above, when  $m\approx 2n$, we are able to retrieve any
solution, no matter sparse or not. However, when $m\approx 1.5n$, we are not able to retrieve general signals.
The point of our numerical experiments is to see if we are able to retrieve sparse signals when $m\approx n$.
The performance of \cref{L1DCalg} is not very good. We need
to improve it.  By using the sparsity, we will enhance the $\ell_1$ based algorithm
by using the projection technique. That is, we project $\bfy^{(k+1)}$ from \cref{L1DCalg} to the set of
all $s$-sparse vectors. That is, we use the hard thresholding technique to update $\bfy^{(k+1)}$.
This leads to an $\ell_1$ DC based algorithm with hard thresholding technique given below.

\begin{algorithm}
\caption{$\ell_1$ DC based Algorithm with Hard Thresholding}
\label{LaisAPA}
\begin{algorithmic}
\STATE{We use the same initialization as in the previous subsection.}
\WHILE{$k\ge 1$}
\STATE{Solve \cref{DCalg4SPR2} using the shrinkage-thresholding technique to get $\bfy^{(k+1)}$.}
\STATE{Apply a modified Attouch-Peypouquet technique  to get a new  $\bfy^{(k+1)}$. }
\STATE{Project $\bfy^{(k+1)}$ into the collection ${\cal R}_s$, $s$-sparse set. That is, let $\bfz^{(k+1)}$ solve the following minimization problem:
\begin{equation}
\label{Sproj}
\sigma_s(\bfx^k)=\min_{\bfz\in {\cal R}_s}\|\bfy^{(k+1)}- \bfz\|_1
\end{equation}}
\STATE{Let $\bfy^{(k+1)}= \bfz^{(k+1)}$}
\ENDWHILE
\RETURN the maximal number of iterations $\bfy^T$
\end{algorithmic}
\end{algorithm}

We now present some numerical results to demonstrate that \cref{LaisAPA} works well.

\begin{example}
 Fix $m\le n$. Many existing computational algorithms fail as the number of measurements is too small.
However, \cref{LaisAPA}
is able to retrieve sparse solutions from the phaseless measurements. Let us present our numerical findings
in \cref{prtab8} and \cref{prtab9}.
\begin{table}[htbp]
{\footnotesize
\caption{The numbers of successful retrieval of sparse solutions with sparsities $s=1, 5, 10, 20, 30, 40$
over 100 repeated runs
based on numbers of measurements satisfying the relations $m/n$ listed below,
where $n=100$ is the size of the solution \label{prtab8}}
\begin{center}
\begin{tabular}{|c|c|cccccccccc|}\hline
$m/n$ &   &   1.1 &  1.2 &   1.3  &  1.4  &   1.5  &  1.6 &   1.7 & 1.8  & 1.9  & 2 \cr \hline
Alg.~\ref{LaisAPA} & $s=1$ & 75  & 83  &  86  &  80  &   88  &  82  &  83 &   87 &   86  &  94 \cr \hline
Alg.~\ref{LaisAPA} & $s=5$ & 54  &   61 &   72 &  63 &  80   & 80   & 74  &  85  &  83  &  87 \cr \hline
Alg.~\ref{LaisAPA} & $s=10$ &  40  &  37  &  54  &  46 &   54 &   70  &   64  &  72   & 79   & 80\cr \hline
Alg.~\ref{LaisAPA} & $s=20$ & 15  &  16  &  22  &   27  &   36  &  45  &  42  &   47   &  59 &  57\cr \hline
Alg.~\ref{LaisAPA} & $s=30$ &  0  &   3  &   3  &  10   &  15  &  22  &  33   &  36  &  39  &  44\cr \hline
Alg.~\ref{LaisAPA} & $s=40$ & 0  &   0   &   0  &   0   &  6   &  10  &  11   &  17  &  30  & 39 \cr \hline
\end{tabular}
\end{center}
}
\end{table}

\begin{table}[htbp]
{\footnotesize
\caption{The numbers of successful retrieval of sparse solutions with sparsities $s=1,2,...,10$
over 100 repeated runs
based on numbers of measurements satisfying the relations $m/n$ listed above,
where $n=100$ is the size of the solution \label{prtab9}}
\begin{center}
\begin{tabular}{|c|c|cccccc|}\hline
$m/n$ &   &      1  &  0.9 &   0.8 & 0.7  & 0.6  & 0.5 \cr \hline
Alg.~\ref{LaisAPA} & $s=1$ &   83  &   80 &   66 &   66 &   58 &   52 \cr \hline
Alg.~\ref{LaisAPA} & $s=2$  &  80  &  64  &  73  &  57  &  53  &  45\cr \hline
Alg.~\ref{LaisAPA} & $s=4$  &  59  &  60  &  47  &  43  &  30  &  18\cr \hline
Alg.~\ref{LaisAPA} & $s=5$   &  56  &  44  &  38  &  26  &  11  &  7\cr \hline
Alg.~\ref{LaisAPA} & $s=10$  &  23 &   14  &   4 &    0 &    0 &  0\cr \hline
\end{tabular}
\end{center}
}
\end{table}

From  \cref{prtab8} and \cref{prtab9}, we can see that \cref{LaisAPA} is able to recover
sparse solutions with high frequency of success.
\end{example}

\section{Appendix}
\label{sec:appendix}
In this section we give some deterministic description of the landscape function of the minimizing function $F$
in \cref{LSmin}. For convenience,
let $A_\ell= \bfa_\ell \bar{\bfa}_\ell^\top$ be the Hermitian matrix of rank one for $\ell=1, \cdots, m$. We first need
\begin{definition}
We say ${\bf a}_j, j=1, \cdots, m$ are  $a_0$-generic if they satisfy
\begin{equation*}
 \|(\bfa_{j_1}^*\bfy, \ldots, \bfa_{j_n}^*\bfy)\|\ge a_0 \|\bfy\|, \quad
\forall \bfy\in \mathbb{C}^n
\end{equation*}
for a positive $a_0\in (0,1)$ for any $1\le j_1<j_2 < \cdots < j_n\le m$.
\end{definition}

\begin{theorem}
\label{mjlai07062018}
Consider the real variable setting.
Let $H_f(\bfx)$ be the Hessian of the minimizing function $f(\bfz)$ and let
$\bfx^\star$ be a global minimizer of \cref{LSmin}. Suppose that ${\bf a}_j, j=1, \cdots, m$ are in $a_0$-generic position.
Then $H_f(\bfx^\star)$ is positive definite.
\end{theorem}
\begin{proof}
Recall $A_\ell =\bfa_\ell \bar{\bfa}_\ell^\top$ for $\ell=1, \cdots, m$.
It is easy to see
$$\nabla f(\bfx)=  2\sum_{\ell=1}^m (\bfx^\top A_\ell\bfx- b_\ell) A_\ell \bfx$$
and the entries
$h_{ij}$ of the Hessian $H_f(\bfx)$ is
\begin{align*}
h_{ij} &= \frac{\partial}{\partial x_i}\frac{\partial}{\partial x_j} f(\bfx)
= 2\sum_{\ell=1}^m (\bfx^\top A_\ell\bfx- b_\ell) a_{ij}(\ell) + 4 \sum_{p=1}^n a_{i,p}(\ell) x_p
\sum_{q=1}^n a_{j,q}(\ell)x_q,
\end{align*}
where $A_\ell= [a_{ij}(\ell)]_{ij=1}^n$.
As we have $(\bfx^*)^\top A_\ell\bfx^*= b_\ell, \ell=1, \cdots, m$, the first summation
term of $h_{ij}$ above is zero at $\bfx^*$.
Letting $M(\bfy) = \bfy^\top H_f(\bfx^*) \bfy$ be a quadratic function of $\bfy$, we have
\begin{align*}
M(\bfy) &= 4 \sum_{\ell=1}^m (\bfy^\top A_\ell \bfx^* (\bfx^*)^\top A_\ell \bfy =
4\sum_{\ell=1}^m |\bfy^\top A_\ell \bfx^*|^2\cr
&= 4 \sum_{\ell=1}^m |\bfy^\top \bfa_\ell|^2 |\bar{\bfa}_\ell^\top \bfx^*|^2
\ge 4a_0  \|\bfx^*\|^2  \|\bfy\|^2.
%\ge \min_{\bar{\bfa}_\ell^\top \bfx^*\not=0}|\bar{\bfa}_\ell^\top \bfx^*|^2   4(1-\delta) \|\bfy\|^2.
\end{align*}
by using the definition of $a_0$ as in a previous section.
It follows that $H_f(\bfx^*)$ is positive definite.
\end{proof}

Next let us show that the global minimizer $\bfx^\star$ in the complex setting. In this case, the Hessian $H_F(\bfx^*)$ is no
longer positive definite. Instead, it is nonnegative definite. To this end, let us fix some notations. Write
$\bfa_\ell = a_\ell + \bfi c_\ell$ for $\ell=1, \cdots, m$. For $\bfz= \bfx+ \bfi
\bfy$, we have $\bfa_\ell^\top \bfz^*=  b_\ell$ for the global minimizer $\bfz^*$.
Writing $f_\ell(\bfx, \bfy) = |\bfa_\ell^\top \bfz|^2 -b_\ell
=(a_\ell^\top \bfx - c_\ell^\top \bfy)^2 + (c_\ell^\top \bfx+ a_\ell^\top \bfy)^2 -b_\ell$, we consider
\begin{equation}
\label{complexFun}
f(\bfx, \bfy) =  \frac{1}{m} \sum_{\ell=1}^m f_\ell^2.
\end{equation}
The gradient  of $f$ can be easily found as follows:  $\nabla f= [\nabla_\bfx f, \nabla_\bfy f]$ with
\begin{align}
\label{gradx}
\nabla_\bfx f(\bfx, \bfy) =& \frac{1}{m}\sum_{\ell=1}^m \nabla_\bfx f_\ell^2
=  \frac{4}{m}\sum_{\ell=1}^m f_\ell(\bfx, \bfy) [(a_\ell^\top \bfx- c_\ell^\top \bfy) a_\ell
+ (c_\ell^\top \bfx +  a_\ell^\top \bfy) c_\ell]
\end{align}
and
\begin{align}
\label{grady}
\nabla_\bfy f(\bfx, \bfy) =& \frac{1}{m}\sum_{\ell=1}^m \nabla_\bfy f_\ell^2
=  \frac{4}{m}\sum_{\ell=1}^m f_\ell(\bfx, \bfy) [(a_\ell^\top \bfx- c_\ell^\top \bfy) (-c_\ell)
+ (c_\ell^\top \bfx +  a_\ell^\top \bfy) a_\ell].
\end{align}
The Hessian of $f$ is more complicated:
\begin{equation}
\label{HessianF}
H_f (\bfx, \bfy) = \left[ \begin{matrix} \nabla_{\bfx}\nabla_\bfx f(\bfx, \bfy) & \nabla_{\bfx}\nabla_\bfy f(\bfx, \bfy)\cr
\nabla_{\bfy}\nabla_\bfx f(\bfx, \bfy ) & \nabla_{\bfy}\nabla_\bfy f(\bfx, \bfy)\end{matrix}\right]
\end{equation}
with $\nabla_{\bfx}\nabla_\bfx f(\bfx, \bfy), \cdots, \nabla_{\bfy}\nabla_\bfy f(\bfx, \bfy)$ given below.
\begin{align*}
\nabla_{\bfx}\nabla_\bfx f(\bfx, \bfy) &=
\frac{4}{m}\sum_{\ell=1}^m f_\ell(\bfx, \bfy) [a_\ell a_\ell^\top  +  c_\ell c_\ell^\top] \cr
&+ \frac{8}{m}\sum_{\ell=1}^m [(a_\ell^\top \bfx- c_\ell^\top \bfy) a_\ell
+ (c_\ell^\top \bfx +  a_\ell^\top \bfy) c_\ell]
[(a_\ell^\top \bfx- c_\ell^\top \bfy) a_\ell^\top
+ (c_\ell^\top \bfx +  a_\ell^\top \bfy) c_\ell^\top]
\end{align*}
\begin{align*}
\nabla_{\bfx}\nabla_\bfy f(\bfx, \bfy) &=
\frac{4}{m}\sum_{\ell=1}^m f_\ell(\bfx, \bfy) [a_\ell (-c_\ell)^\top  +  c_\ell a_\ell^\top] \cr
&+ \frac{8}{m}\sum_{\ell=1}^m [(a_\ell^\top \bfx- c_\ell^\top \bfy) a_\ell+ (c_\ell^\top \bfx +  a_\ell^\top \bfy) c_\ell]
[(a_\ell^\top \bfx- c_\ell^\top \bfy) (-c_\ell)^\top
+ (c_\ell^\top \bfx +  a_\ell^\top \bfy) a_\ell^\top]
\end{align*}
and similar for $\nabla_{\bfy}\nabla_\bfx f(\bfx, \bfy)$ and $\nabla_{\bfy}\nabla_\bfy f(\bfx, \bfy)$.

We are now ready to prove the nonnegativity of the Hessian at a global minimizer $\bfz^*$ in the complex value setting.
\begin{theorem}
\label{Xu7192018}
At any global minimizer $\bfz^*=(\bfx^*, \bfy^*)$, we have the Hessian $H_f(\bfx^*, \bfy^*)\ge 0$. In fact,
$H_f(\bfx^*,\bfy^*)=0$ along the direction $[-(\bfy^*)^\top, (\bfx^*)^\top ]^\top$.
\end{theorem}
\begin{proof}
At $\bfz^*=\bfx^*+\bfi \bfy^*$, we have
\begin{align*}
\nabla_{\bfx}\nabla_\bfx f(\bfx^*, \bfy^*) &=  \frac{8}{m}\sum_{\ell=1}^m [(a_\ell^\top \bfx^*- c_\ell^\top \bfy^*) a_\ell
+ (c_\ell^\top \bfx^* +  a_\ell^\top \bfy^*) c_\ell] \,
[(a_\ell^\top \bfx^*- c_\ell^\top \bfy^*) a_\ell^\top
+ (c_\ell^\top \bfx^* +  a_\ell^\top \bfy^*) c_\ell^\top]
\end{align*}
\begin{align*}
\nabla_{\bfx}\nabla_\bfy f(\bfx^*, \bfy^*) &=
 \frac{8}{m}\sum_{\ell=1}^m [(a_\ell^\top \bfx^* - c_\ell^\top \bfy^*) a_\ell+ (c_\ell^\top \bfx^*
+  a_\ell^\top \bfy^*) c_\ell] \,
[(a_\ell^\top \bfx^* - c_\ell^\top \bfy^*) (-c_\ell)^\top
+ (c_\ell^\top \bfx^* +  a_\ell^\top \bfy^*) a_\ell^\top]
\end{align*}
and similar for the other two terms. It is easy to see that for any $\bfw= \bfu +\bfi \bfv$ with $\bfu, \bfv\in \mathbb{R}^n$,
we have
\begin{align*}
&[\bfu^\top \, \bfv^\top]^\top H_f(\bfx^*, \bfy^*) \left[\begin{matrix}\bfu \cr \bfv\end{matrix}\right]\cr
=& \frac{8}{m}\sum_{\ell=1}^m [(a_\ell^\top \bfx^* - c_\ell^\top \bfy^*) a_\ell^\top \bfu + (c_\ell^\top \bfx^*
+  a_\ell^\top \bfy^*) c_\ell^\top \bfu ]^2 \cr
&+ \frac{8}{m}\sum_{\ell=1}^m 2[(a_\ell^\top \bfx^* - c_\ell^\top \bfy^*) a_\ell^\top \bfu + (c_\ell^\top \bfx^*
+  a_\ell^\top \bfy^*) c_\ell^\top \bfu ] \,
[(a_\ell^\top \bfx^* - c_\ell^\top \bfy^*) (-c_\ell)^\top\bfv + (c_\ell^\top \bfx^* +  a_\ell^\top \bfy^*) a_\ell^\top\bfv]\cr
&+ \frac{8}{m}\sum_{\ell=1}^m
[(a_\ell^\top \bfx^* - c_\ell^\top \bfy^*) (-c_\ell)^\top\bfv + (c_\ell^\top \bfx^* +  a_\ell^\top \bfy^*) a_\ell^\top\bfv]^2\cr
=& \frac{8}{m}\sum_{\ell=1}^m [(a_\ell^\top \bfx^* - c_\ell^\top \bfy^*) a_\ell^\top \bfu + (c_\ell^\top \bfx^*
+  a_\ell^\top \bfy^*) c_\ell^\top \bfu
 +(a_\ell^\top \bfx^* - c_\ell^\top \bfy^*) (-c_\ell)^\top\bfv + (c_\ell^\top \bfx^* +  a_\ell^\top \bfy^*) a_\ell^\top\bfv]^2\cr
\ge & 0.
\end{align*}
Furthermore, if we choose $\bfu= -\bfy^*$ and $\bfv=\bfx^*$, we can easily see that the Hessian $H_f$
along this direction is zero:
$$
[-(\bfy^*)^\top \, (\bfx^*)^\top]^\top H_f(\bfx^*, \bfy^*) \left[\begin{matrix}-\bfy^* \cr \bfx^*\end{matrix}\right]=0.
$$
\end{proof}

\bibliographystyle{siamplain}

\end{document}